\documentclass[USenglish,oneside,twocolumn]{article}

\usepackage[utf8]{inputenc}
\usepackage[big]{dgruyter_NEW}
 
\usepackage{amsmath, amssymb, amsthm}
\usepackage{hyperref}
\usepackage{url}
\usepackage{xcolor}
\usepackage{color}
\usepackage{balance}
\usepackage{graphicx}
\usepackage{subfigure} 
\usepackage{caption}
\usepackage{natbib}
\usepackage{enumitem}
\usepackage{balance}
\usepackage{float}
\usepackage{icomma}
\usepackage{tabularx}
\usepackage{mathtools}

\DeclareMathOperator*{\argmin}{arg\,min}

\usepackage[algoruled,vlined,nofillcomment,linesnumbered]{algorithm2e}

\theoremstyle{plain}
\newtheorem{thm}{Theorem}[section]

\newtheorem{lem}[thm]{Lemma}

\newtheorem{cor}[thm]{Corollary}
\newtheorem{defn}[thm]{Definition}

\newcommand{\V}{\mathbb{V}}
\newcommand{\Expect}{{\rm I\kern-.3em E}}
\DeclareMathOperator*{\E}{\mathbb{E}} 
 
\makeatletter
\renewcommand\huge{\@setfontsize\huge{19.97}{26}}
\makeatother 

\makeatletter
\def\blfootnote{\xdef\@thefnmark{}\@footnotetext}
\makeatother

\cclogo{\includegraphics{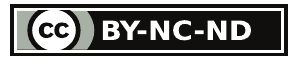}}
  
\begin{document}
 
  \author[1]{Brendan Avent}

  \author[2]{Yatharth Dubey}

  \author[3]{Aleksandra Korolova}

  \affil[1]{University of Southern California, E-mail: bavent@usc.edu}

  \affil[2]{Georgia Institute of Technology$^\dagger$, E-mail: yatharthdubey7@gatech.edu}

  \affil[3]{University of Southern California, E-mail: korolova@usc.edu}
  
  \title{\huge The Power of the Hybrid Model for Mean Estimation}

  \runningtitle{The Power of the Hybrid Model for Mean Estimation}

\begin{abstract}
{We explore the power of the hybrid model of differential privacy (DP), in which some users desire the guarantees of the local model of DP and others are content with receiving the trusted-curator model guarantees. In particular, we study the utility of hybrid model estimators that compute the mean of arbitrary real-valued distributions with bounded support. When the curator knows the distribution's variance, we design a hybrid estimator that, for realistic datasets and parameter settings, achieves a constant factor improvement over natural baselines. We then analytically characterize how the estimator's utility is parameterized by the problem setting and parameter choices. When the distribution's variance is unknown, we design a heuristic hybrid estimator and analyze how it compares to the baselines. We find that it often performs better than the baselines, and sometimes almost as well as the known-variance estimator. We then answer the question of how our estimator's utility is affected when users' data are not drawn from the same distribution, but rather from distributions dependent on their trust model preference. Concretely, we examine the implications of the two groups' distributions diverging and show that in some cases, our estimators maintain fairly high utility. We then demonstrate how our hybrid estimator can be incorporated as a sub-component in more complex, higher-dimensional applications. Finally, we propose a new privacy amplification notion for the hybrid model that emerges due to interaction between the groups, and derive corresponding amplification results for our hybrid estimators.}
\end{abstract}

  \journalname{Proceedings on Privacy Enhancing Technologies}
  \startpage{1}
%

\maketitle

\blfootnote{This paper appears in the 2020 Proceedings on Privacy Enhancing Technologies.}
\blfootnote{$^\dagger$ Work done while visiting University of Southern California.}

\section{Introduction}
Differential privacy~\citep{dmns06}, has become one of the de facto standards of privacy in computer science literature, particularly for privacy-preserving statistical data analysis and machine learning.
Two traditional models of trust in DP literature are: the trusted-curator model (TCM) and the local model (LM).
In the TCM, the curator receives the users' true data and applies a randomized perturbation to achieve DP.
In the LM, the curator receives users' privatized data through a locally randomizing oracle that individually ensures DP for each user.

When it comes to deployments of DP, curators (e.g., companies, social scientists, government organizations) and users alike find the LM to be a better match for their privacy goals~\citep{wired, pew2015privacy}.
Users' privacy is assured even when they don't trust the curator, and the curator limits its liability in the face of data leaks.
However, it is well understood theoretically and empirically that utility outcomes are far worse in the LM than in the TCM~\citep{kairouz2014, bassily2015local, duchi2018minimax, bittau2017prochlo, fanti2016building}.
This poses a challenge for curators with smaller user bases than the tech giants -- on the one hand, they want to guarantee local DP to their users; on the other hand, they won't be able to gain much utility from the data if they do.

Until recently, these trust models were considered mutually exclusively.
Recent work of Avent et al.~\citep{blender} observed that it can be beneficial to consider a \textit{hybrid model} in which the majority of the users desire privacy in the LM, but a small fraction of users are willing to contribute their data with TCM guarantees.
Indeed, it is common in industry to have a small group of ``early adopters'' or ``opt-in users'' who are willing to trust the organization more than the average user~\citep{microsoft-opt-in}.
The work of~\citep{blender} demonstrated experimentally that in the hybrid model, one can develop algorithms that take advantage of the opt-in user data to improve utility for the task of local search.
However, their results left open the questions of how much improvement can be gained compared with the LM, the dependence of improvement on the parameters (e.g., sample size, number of opt-in users, privacy level, etc.), and whether hybrid algorithms exist that improve over both the TCM and LM algorithms simultaneously for all parameters (as their proposed algorithm, BLENDER, was only able to achieve simultaneous improvement for \textit{some} parameters).
These are precisely the questions we address in this work for the problem of mean estimation for bounded real-valued distributions -- a well-studied problem in statistical literature due to its prevalence as a fundamental building block in solutions to more complex tasks.\\

\noindent \textbf{Contributions: } Our contributions are as follows.
\vspace{-1em}
\begin{itemize}[leftmargin=*]
\item We initiate the study of mean estimation in the hybrid model, where users with bounded real-valued data self-partition into two groups based on their preferred trust model. We rigorously formalize this problem in a statistical framework (Section~\ref{sec:prelim}), making minimal distributional assumptions for user data and even allowing the groups to come from separate distributions.
\item We define a family of hybrid estimators that utilize a generic class of DP mechanisms (Section~\ref{sec:estimators}). To evaluate the hybrid estimators' relative quality, we detail two non-hybrid baseline estimators and theoretically analyze their relationship.
\item When the groups have the same distribution and the curator knows its variance, we derive a hybrid estimator from this family and analytically quantify its utility (Section~\ref{sec:homo-kv}). First, we prove that it always outperforms both non-hybrid baselines. Second, we prove that for practical parameters, it outperforms both baselines by a factor of no greater than $2.286$. Additionally, we empirically evaluate our hybrid estimator on realistic distributions, showing that it achieves high utility in practice.
\item When the groups have the same distribution but the curator \emph{doesn't} know its variance, we derive another hybrid estimator from this family and analytically quantify the estimator's utility (Section~\ref{sec:homo-uv}). We prove that it always outperforms at least one non-hybrid baseline, and we precisely determine the conditions under which it outperforms both. We empirically evaluate it on realistic distributions and find that it not only achieves high utility in practice, but is sometimes utility competitive with the known-variance case.
\item Since users' self-partitioning may induce a bias between the groups, we evaluate our analytic utility expressions in the cases where the groups' distributions diverge (Section~\ref{sec:hetero}). We find that the hybrid estimator is robust to divergences in the variances of the groups' distributions, but sensitive to divergences in the means of the groups' distributions.
\item To demonstrate how more complex algorithms can use our estimator as a sub-component, we design a hybrid $K$-means algorithm which uses the hybrid estimator to merge the intermediate results of two non-hybrid $K$-means algorithms (Section~\ref{sec:apps}). We experimentally show that this algorithm is able to achieve utility on-par with the better of its two non-hybrid building blocks, even though its underlying hybrid estimator is not explicitly designed for this problem.
\item We introduce a new privacy amplification notion for the hybrid model that stems from interaction between the groups (Section~\ref{sec:amp}). We derive the amplification level that our hybrid estimator achieves, and show that this amplification is significant in practice.
\end{itemize}

\section{Preliminaries} \label{sec:prelim}
In this section, we present the requisite background on differential privacy, define the mean estimation problem setting, and then review related work.

\subsection{Differential Privacy Background}
In this background, we precisely define differential privacy, then describe two of the most popular DP mechanisms, and conclude with a discussion of trust models.

Formally, a mechanism $\mathcal{M}$ is $(\epsilon, \delta)$-DP~\cite{dmns06} if and only if for all neighboring databases $D$ and $D^\prime$ differing in precisely one user's data, the following inequality is satisfied for all possible sets of outputs $Y \subseteq Range(\mathcal{M})$:
\vspace{1mm}
$$\Pr[\mathcal{M}(D) \in Y] \le e^\epsilon \Pr[\mathcal{M}(D^\prime) \in Y] + \delta.$$
\vspace{-3mm}

\noindent A mechanism that satisfies $(\epsilon, 0)$-DP is said to be $\epsilon$-DP.

Two of the most popular DP mechanisms are the Laplace mechanism~\cite{dmns06} and the Gaussian mechanism~\cite{dwork2006our}.
These mechanisms ensure DP for any dataset $D$ evaluated under a real-valued function $f$ by computing $\tilde{f}(D) = f(D) + Y$.
For the Laplace mechanism, $Y$ is a random variable drawn from the Laplace distribution with scale parameter $b=\Delta_1 f / \epsilon$ (yielding standard deviation $s = \sqrt{2}b$), and $\Delta_1 f = \max \left\Vert f(D) - f(D^\prime) \right\Vert_1$ over all neighboring $D, D^\prime$.
For the Gaussian mechanism, $Y$ is drawn from the Gaussian distribution with standard deviation $s = \sqrt{2\ln(1.25/\delta)}\Delta_2 f / \epsilon$, and $\Delta_2 f = \max \left\Vert f(D) - f(D^\prime) \right\Vert_2$ over all neighboring $D, D^\prime$.

As discussed in the introduction, there are two classic trust models in DP, distinguished by their timing of when the privacy perturbation is applied.
In the LM, user data undergoes a privacy-preserving perturbation before it is sent to the curator; in the TCM the curator first collects all the data, and then applies a privacy-preserving perturbation.
The hybrid model, first proposed in~\citep{blender}, enables algorithms to utilize a combination of trust models.
Specifically, the hybrid model allows users to individually select between the TCM and LM based on their personal trust preferences.

\subsection{Problem Setting} \label{sec:setting}
Statistical literature on mean estimation spans a wide range of assumptions and utility objectives, so we begin by stating ours.

There are $n$ users, with each user $i \in [n]$ holding data $x_i$ to be used in a differentially private computation.
Users self-partition into the TCM or the LM group and, regardless of their group choice, are guaranteed the same level of DP.
Thus, a user's group choice only reflects their trust towards the curator.
The fraction of users that opted-in to the TCM is denoted as $c \in (0, 1)$, while the remaining $(1-c)$ fraction prefer the LM.
We denote the two groups as indicies in the sets $T = \{1,\dots, cn\}$ and $L = \{cn+1,\dots, n\}$ respectively, such that $T \cup L = [n]$.

Users who opt-in to the TCM group (referred to as TCM users) have data $x_i$ drawn iid from an unknown distribution $\mathcal{D}_T$ with mean $\mu_T$, variance $\sigma^2_T$, and support on the subset of interval $[0, m_T]$.
Users who chose the LM group (referred to as LM users) have data $x_i$ drawn iid from an unknown distribution $\mathcal{D}_L$ with mean $\mu_L$, variance $\sigma^2_L$, and support on the subset of interval $[0, m_L]$.
Together, the groups' distributions form a mixture distribution $\mathcal{D} = c\mathcal{D}_T + (1-c)\mathcal{D}_L$ with mean $\mu = c\mu_T + (1-c)\mu_L$, variance $\sigma^2~=~c(\mu_T^2 + \sigma_T^2) + (1-c)(\mu_L^2 + \sigma_L^2)$, and support on $[0, m]$ where $m = \max\{m_T, m_L\}$.
Table~\ref{tab:notation_table} provides a summary of all notation introduced in this work.

We make minimal assumptions about these distributions, and the curator's knowledge thereof, throughout the paper.
Specifically, in Sections~\ref{sec:homo-kv} and~\ref{sec:homo-uv}, we assume $\mathcal{D} = \mathcal{D}_T = \mathcal{D}_L$ and analyze the scenarios where the curator both does and doesn't know $\mathcal{D}$'s variance respectively.
In Section~\ref{sec:hetero}, we lift this equal-distributions assumption and analyze the consequences of the groups' distributions diverging.

\vspace{-1em}
\paragraph*{Measuring Utility}
Our goal is to design accurate estimators of the mean $\mu$ of the mixture distribution $\mathcal{D}$.
To measure utility, we benchmark all estimators against the non-private empirical mean estimator.
\begin{defn} \label{def:empirical}
The non-private empirical mean estimator is:
$$\hat{\mu} = \frac{1}{n} \sum_{i \in [n]} x_i = c\hat{\mu}_T + (1-c)\hat{\mu}_L.$$
\end{defn}

This choice of benchmark reflects the fact that we are interested in the \emph{excess error} introduced by the privatization scheme, beyond the inherent error induced by a finite sample size.
Concretely, we measure the absolute error of an estimator $\tilde{\mu}$ by explicitly computing the mean squared error between it and the empirical mean.
\begin{defn} \label{def:mse}
The MSE between an estimator $\tilde{\mu}$ and the non-private empirical mean $\hat{\mu}$ is:
$$\mathcal{E} = \text{MSE}(\tilde{\mu}, \hat{\mu}) = \E[(\tilde{\mu} - \hat{\mu})^2]$$
\end{defn}
\noindent Since the non-private empirical benchmark is used to measure the MSEs of all estimators in this paper, we simply refer to it as the MSE of the estimator.

\begin{table}[h]
\vspace{1em}
\setlength\tabcolsep{2pt}
\begin{tabularx}{\columnwidth}{c|X}
\toprule
Symbol & Usage \\
\midrule
$\epsilon, \delta$ & Differential privacy parameters \\
\hline
$n$ & Total number of users \\ 
$c$ & Fraction of users who opt-in to TCM \\
$T, L$ & Set of users who opted-in to TCM and set of users who are using LM, respectively \\
\hline
$\mathcal{D}$ & Mixture distribution of both groups' data \\
$\mu, \sigma^2, m$ & Mean, variance, and maximum support of $\mathcal{D}$ \\
$\mathcal{D}_T$ & Distribution of TCM groups' data \\
$\mu_T, \sigma^2_T, m_T$ & Mean, variance, and maximum support of $\mathcal{D}_T$ \\
$\mathcal{D}_L$ & Distribution of LM groups' data \\
$\mu_L, \sigma^2_L, m_L$ & Mean, variance, and maximum support of $\mathcal{D}_L$ \\
$x_i$ & User $i$'s private data drawn iid from its group's distribution \\
\hline
$\hat{\mu}, \hat{\mu}_T, \hat{\mu}_L$ & Empirical mean estimates with all users, with only the TCM users, and with only the LM users, respectively \\
$\mathcal{E}$ & MSE of an estimator with respect to $\hat{\mu}$ \\
\hline
$\tilde{\mu}_T, \mathcal{E}_{T}$ & TCM-Only estimator and its MSE \\
$\tilde{\mu}_F, \mathcal{E}_{F}$ & Full-LM estimator and its MSE \\
$\tilde{\mu}_L, \mathcal{E}_{L}$ & LM-Only estimator and its MSE \\
$\tilde{\mu}_{H}(w), \mathcal{E}_{H}(w)$ & Hybrid estimator with weight $w$ and its MSE \\
\hline
$Y_T, s_T^2$ & TCM-Only estimator's privacy random variable and its variance \\
$Y_{L,i}, s_L^2$ & User $i$'s local privacy random variable and its variance \\
\hline
$n_{crit}, c_{crit}$ & $n$ and $c$ values that partition where $\mathcal{E}_{T} \le \mathcal{E}_{F}$ \\
$R(\mathcal{E})$, $r(\mathcal{E})$ & Relative improvement of estimator with MSE $\mathcal{E}$ over the best and worst non-hybrid baselines, respectively \\
\bottomrule
\end{tabularx}
\caption{Comprehensive list of notation.} \label{tab:notation_table}
\vspace{-2em}
\end{table}

\subsection{Related Work} \label{sec:related-work}
We first compare our paper to the closest related work in the hybrid model~\cite{blender}, then discuss other works on DP mean estimation in non-hybrid models, and conclude by discussing other work in hybrid trust models.

\paragraph*{Comparison to BLENDER~\cite{blender}}
A shared goal of our work and~\cite{blender} is to take advantage of the hybrid model; beyond that, our work is fundamentally different from theirs in several ways.

The works address different problems.
Avent et al. studied the problem of local search, which is a specific problem instance of heavy-hitter identification and frequency estimation.
BLENDER tackles the frequency estimation portion of the problem by estimating counts of boolean-valued data using a variant of randomized response~\cite{warner1965randomized}.
Our work focuses on the conceptually simpler, but not strictly weaker, problem of mean estimation of real-valued data using a broad class of privatization mechanisms.
Because of this, their methods aren't applicable in this work.

Both works compare against the same types of baselines in their respective problems, but reach very different conclusions.
The baselines are: 1) using only the TCM group's data under the TCM, and 2) using all data under the LM.
\cite{blender} experimentally evaluated BLENDER and found that it typically outperformed at least one of these baselines, and occasionally outperformed both.
For our problem, we derive utility expressions which prove that not only does our estimator always outperform at least one of the baselines, but that under certain assumptions, it \emph{always outperforms both}.

Since the hybrid model enables users to self-partition into groups based on their trust model preference, an important consideration for utility is whether the groups have the same data distribution.
In BLENDER, it was assumed that they did.
In this work, our setting allows for groups to have the same or different distributions, and we derive analytic results for both cases.

Finally, the works have different takes on the role of interaction between groups.
BLENDER carefully utilizes inter-group interactivity to achieve high utility.
In this work, our hybrid estimators have no inter-group interactivity; these estimators achieve high utility, demonstrating that such interactivity isn't always necessary for improving utility.
Moreover, we find that our lack of interactivity can improve users' privacy guarantees with respect to a specific type of adversary, whereas BLENDER's interactivity gives no such improvement.

\paragraph*{Non-Hybrid Mean Estimation}
In this work, we use simple non-hybrid baseline mean estimators to enable us to obtain exact finite-sample utility expressions. 
However, DP mean estimation of distributions under both the TCM and LM has been studied since the models' introductions~\cite{blum2005practical, warner1965randomized, duchi2013local}, and continues to be actively studied to this day~\cite{feldman2016dealing, karwa2017finite, acharya2018inspectre, gaboardi2018locally, duchi2018minimax, kamath2018privately, kamath2019differentially, joseph2019locally, du2020differentially, bun2019average, kamath2020private, ghazi2020private, biswas2020coinpress}.
The goal of mean estimation research under both models is to maximize utility while minimizing the sample complexity by making various distributional assumptions.
Some assumptions are stronger than those made in this work, such as assuming the data is drawn from a narrow family of distributions.
Other assumptions are weaker, such as requiring only that the mean lies within a certain range or that higher moments are bounded.
Because of the complexity of the mechanisms and their reliance on the distributional assumptions in the related works, their utility expressions are typically bounds or asymptotic rather than exact.
Since we need exact finite-sample utility expressions to precisely determine the utility of our hybrid estimator relative to the baselines, we are unable to use their estimators and assumptions.
Nevertheless, the related works show a practically significant sample complexity gap between the TCM and LM in their respective settings, further motivating mean estimation in the hybrid model.

\paragraph*{Other Works in Hybrid Trust Models}
Several other works utilize a hybrid combination of trust models.
Of these, the closest-related work is the concurrent work of Beimel et al.~\cite{beimel2020power}.
Their work examines precisely the same hybrid DP model as this work, the combined \emph{trusted-curator/local} model, and has the same goal of understanding whether this hybrid model is more powerful than its composing models.
To accomplish this goal, they perform mathematical analyses on several theoretical problems, deriving asymptotic bounds which show that it is possible to solve problems in the hybrid model which cannot be solved in the TCM or LM separately. 
Additionally, they show that there are problems which cannot be solved in the TCM or LM separately, and can be solved in the hybrid model, but only if the TCM and LM groups interact with each other.
Finally, they analyze a problem which \emph{does not} significantly benefit from the hybrid model: basic hypothesis testing.
They prove that if there exists a hybrid model mechanism that distinguishes between two distributions effectively, then there also exists a TCM or LM mechanism which does so nearly as effectively. This result implies a lack of power of the hybrid model for the problem of mean estimation in certain settings.

Beyond the trusted-curator/local hybrid model, there are multiple alternative hybrid models in DP literature.
The most popular is the \emph{public/private} hybrid model of Beimel et al.~\cite{beimel2013private} and Zhanglong et al.~\cite{ji2013differential}.
In this model, most users desire the guarantees of differential privacy, but some users have made their data available for use without requiring any privacy guarantees.
In this model, some works assume that DP is achieved in the TCM~\cite{hamm2016learning, papernot2016semi}, while others assume that DP is achieved in the LM~\cite{xiong2016randomized, wang2019estimating}.
In both cases, the works show that by operating in the public/private hybrid model, one can significantly improve utility relative to either model separately.
Recently, theoretical works~\cite{bassily2019limits, bassily2020private} have explored the limits of this model's power via lower bounds on the sample complexity of fundamental statistical problems.

Another DP hybrid model recently introduced is the \emph{shuffle} model, which was conceptually proposed by Bittau et al.~\cite{bittau2017prochlo} before being mathematically defined and analyzed for its DP guarantees by Cheu et al.~\cite{cheu2019distributed} and Erlingsson et al.~\cite{erlingsson2019amplification}.
In this model, users privately submit their data under the LM via an anonymous channel to the curator.
The anonymous channel randomly permutes the users' contributions so that the curator has no knowledge of what data belongs to which user.
This ``shuffling'' enables users to achieve improved DP guarantees over their LM guarantees in isolation.
Several works have since improved the original analyses and expanded the shuffle model to achieve even greater improvements in the users' DP guarantee~\cite{balle2019privacy, ghazi2019scalable, ghazi2019power, balle2019differentially, ghazi2020pure, ghazi2020private}.

\section{DP Estimators} \label{sec:estimators}
In this section, we introduce the baseline estimators in the classic DP models, describe how we compare new estimators against these baselines, and define the family of hybrid estimators that we will be working with.

\subsection{Baseline Non-hybrid DP Estimators}
To understand the utility of the hybrid model, we put it into context with the utility of non-hybrid approaches.
The most natural non-hybrid alternatives are: 1) only using the TCM group's data under the TCM, and 2) using all the data under the LM.
This is motivated directly by the decision that an analyst must make when choosing between these two models: 1) use only the data of the more-trusting users under the TCM so as to not violate the trust preferences of the remaining users, or 2) treat all users the same under the less-trusting LM.

For both baselines, we consider estimators which directly compute the empirical mean, then add $0$-mean noise from an arbitrary distribution whose variance is calibrated to ensure DP under the respective model.
For pure $\epsilon$-DP, this typically corresponds to using the Laplace mechanism; for $(\epsilon, \delta)$-DP, this typically corresponds to using the Gaussian mechanism~\cite{dwork2006our}.
We derive all results for the generic noise-addition mechanisms, and we use the $\epsilon$-DP Laplace mechanism for empirical evaluations.

\paragraph*{TCM-Only Estimator}
The stated consequence of using the TCM is that the LM group's data cannot be used.
Thus, we design an estimator for this model and refer to it as the ``TCM-Only'' estimator.
\begin{defn} \label{def:tcm-only}
The TCM-Only estimator is:
$$\tilde{\mu}_T = \frac{1}{cn} \sum_{i \in T} x_i + Y_T,$$
where $Y_T$ is a random variable with $0$ mean and $s_T^2$ variance chosen such that DP is satisfied for all TCM users.
\end{defn}

\begin{lem} \label{lem:tcm-only-mse}
$\tilde{\mu}_T$ has expected squared error:
$$\mathcal{E}_T = \frac{(1-c)^2}{cn}\sigma_T^2 + \frac{1-c}{n}\sigma_L^2 + s_T^2 + (\mu_T-\mu)^2.$$
\end{lem}
\begin{proof}\renewcommand{\qedsymbol}{}
See Appendix~\ref{proof:tcm-only-mse}.
\end{proof}
\noindent This error has three components, $\frac{(1-c)^2}{cn}\sigma_T^2 + \frac{1-c}{n}\sigma_L^2$, $s_T^2$, and $(\mu_T-\mu)^2$.
The first component is the error induced by subsampling only the TCM users -- we refer to this as the \emph{excess sampling error}.
The second component is the error due to DP -- we refer to this as the \emph{privacy error}.
The third component is the \emph{bias error} induced by the groups' means differing.

\paragraph*{Full-LM Estimator}
Since the LM doesn't require trust in the curator, the data of \emph{all} users can be used under this model.
We design an estimator for this model and refer to it as the ``Full-LM'' estimator.
\begin{defn} \label{def:full-lm}
Suppose each user $i$ privately reports their data as $x_i + Y_{L,i}$, where $Y_{L,i}$ is a random variable with $0$ mean and $s_L^2$ variance chosen such that DP is satisfied for user $i$.
The Full-LM estimator is then:
$$\tilde{\mu}_F = \frac{1}{n}\sum_{i \in [n]} (x_i + Y_{L,i}),$$
\end{defn}

\begin{lem} \label{lem:full-lm-mse}
$\tilde{\mu}_F$ has expected squared error:
$$\mathcal{E}_F = \frac{s_L^2}{n}.$$
\end{lem}
\begin{proof}\renewcommand{\qedsymbol}{}
See Appendix~\ref{proof:full-lm-mse}.
\end{proof}
\noindent This error only consists of a single simple component: the privacy error.
Since the entire dataset is used, there is no excess sampling error and no bias error.

\subsection{Utility Over Both Baselines}
While our absolute measure of an estimator's utility is the MSE (discussed in Section~\ref{sec:setting}), we are primarily interested in a hybrid estimator's \textit{relative} gain over the baseline estimators.
Explicitly, given some hybrid estimator with MSE $\mathcal{E}$, we consider the following measure of relative improvement over the baseline estimators.
\begin{defn} \label{def:R}
The relative improvement of an estimator with MSE $\mathcal{E}$ over the best baseline estimator is:
$$R(\mathcal{E}) = \frac{\min\{\mathcal{E}_T, \mathcal{E}_F\}}{\mathcal{E}}.$$
\end{defn}

This measure of relative improvement can be re-written to explicitly consider the regimes where each of the baseline estimators achieves the $\min\{\cdot\}$.
That is, we determine the parameter configurations in which the TCM-only estimator is better/worse than the Full-LM estimator.
Intuitively, we expect that when very few users opt-in to the TCM, the TCM-Only estimator's large excess sampling error will overshadow its smaller privacy error (relative to the Full-LM estimator's privacy error).
This intuition is made precise by considering ``critical values'' of $c$ and $n$ that determine the regimes where each of the estimators yields better utility.
\begin{lem} \label{lem:crit}
Let $n_{crit}$ and $c_{crit}$ be defined as follows.
\begin{align*}\small
    n_{crit}& = \frac{c s_L^2 + (1-c)((1-c)\sigma_T^2 - c\sigma_L^2)}{c((\mu_T - \mu)^2 + s_T^2)} \\
    c_{crit}& =  \begin{cases}
          \frac{\sigma_L^2}{\sigma_L^2 + s_L^2},& \sigma_T = \sigma_L,\\[1em]
          \frac{2\sigma_T^2 - \sigma_L^2 + s_L^2-\sqrt{(\sigma_L^2 - s_L^2)^2 + 4s_L^2\sigma_T^2}}{2(\sigma_T^2 - \sigma_L^2)},& \sigma_T \neq \sigma_L
         \end{cases}
\end{align*}
We have that $\mathcal{E}_T \le \mathcal{E}_F$ if and only if $c > c_{crit} \wedge n \le n_{crit}$.
\end{lem}
\begin{proof}
Directly reduce the system of inequalities constructed by $\mathcal{E}_T \le \mathcal{E}_F$ in conjunction with the regions given by the valid parameter ranges.
This immediately yields the result.
\end{proof}

\noindent This characterization allows us to partition the definition of relative improvement into the behavior of each baseline estimator, re-written as follows.

\begin{defn} \label{def:R-crit}
The relative improvement of an estimator with MSE $\mathcal{E}$ over the best baseline estimator is:
\begin{equation*}
  R(\mathcal{E}) = \frac{1}{\mathcal{E}} \cdot
  \begin{cases}
	\mathcal{E}_T & \text{if $c>c_{crit} \wedge n \le n_{crit}$} \\
	\mathcal{E}_L & \text{otherwise}
  \end{cases}
\end{equation*}
\end{defn}

The behavior of these two cases further depends on the privacy mechanism used, as that dictates $s_T$ and $s_L$.
For example, when using the $\epsilon$-DP Laplace mechanism in the homogeneous setting where both group means are $\mu$ and variances are $\sigma^2$, these definitions of critical values and relative improvement become the following.

\begin{lem} \label{lem:crit-lap}
Adding $\epsilon$-DP Laplace noise for privacy, define $c_{crit}=\frac{\epsilon^2\sigma^2}{2m^2+\epsilon^2\sigma^2}$ and $n_{crit}=\frac{2m^2}{c(2cm^2-(1-c)\epsilon^2\sigma^2)}$.
We have that $\mathcal{E}_T \le \mathcal{E}_F$ if and only if $c > c_{crit} \wedge n \ge n_{crit}$.
\end{lem}

\begin{defn} \label{def:R-crit-pure-lap}
Adding $\epsilon$-DP Laplace noise for privacy, the relative improvement of an estimator with MSE $\mathcal{E}$ over the best baseline estimator is:
\begin{equation*} 
  R(\mathcal{E}) = \frac{1}{\mathcal{E}} \cdot
  \begin{cases}
	\frac{1-c}{cn}\sigma^2 + \frac{2m^2}{c^2 n^2 \epsilon^2} & \text{if $c>c_{crit} \wedge n \ge n_{crit}$} \\
	\frac{2m^2}{n\epsilon^2} & \text{otherwise}
  \end{cases}
\end{equation*}
\end{defn}

\noindent Thus, once the fraction of users opting-in to the TCM is large enough, the TCM-Only estimator has better MSE than the Full-LM estimator.
In all other regimes, the Full-LM estimator has better MSE than the TCM-Only estimator.
This matches the intuition.\\

\noindent Designing a hybrid estimator which outperforms \textit{at least one} of these baselines in \textit{all} regimes (i.e., for all settings of parameters $\mu, \sigma^2, n, c, m$, etc.) is trivial, as is designing a hybrid estimator which outperforms \textit{both} baselines in \textit{some} regimes.
One challenge solved in this work is designing a hybrid estimator which provably outperforms \textit{both} baselines across \textit{all} regimes.

\subsection{Convex Hybrid DP Estimator Family}
Simply described, this family of hybrid estimators has the two groups independently compute their own private estimates of the mean, then directly combines them as a weighted average.
The TCM group's estimator is the TCM-Only estimator.
The LM group's estimator is almost the same as the Full-LM estimator, except now with reports only from the LM users.
We refer to this as the ``LM-Only'' estimator, and briefly detour to define and analyze it.

\begin{defn}
The LM-Only estimator is:
$$\tilde{\mu}_L = \frac{1}{(1-c)n} \sum_{i \in L} (x_i + Y_{L,i}),$$
where, for each $i \in L$, $Y_{L,i}$ is a random variable with $0$ mean and $s_L^2$ variance chosen such that DP is satisfied for user $i$.
\end{defn}

\begin{lem} \label{lem:lm-only-mse}
$\tilde{\mu}_L$, has expected squared error:
$$\mathcal{E}_L = \frac{c^2}{(1-c)n}\sigma_L^2 + \frac{c}{n}\sigma_T^2 + \frac{1}{(1-c)n}s_L^2 + (\mu_L - \mu)^2.$$
\end{lem}
\begin{proof}\renewcommand{\qedsymbol}{}
See Appendix~\ref{proof:lm-only-mse}.
\end{proof}
\noindent This estimator has excess sampling error, privacy error, and bias error.
Since it has strictly greater error than the Full-LM estimator, it is not used as one of the baseline estimators.\\

We now define a family of convexly-weighted hybrid estimators parameterized by weight $w \in [0,1]$, which we will use throughout this paper.
For any $w$, the hybrid estimator computes a convex combination of the independent TCM-Only and LM-Only estimators.
\begin{defn} \label{def:hybrid}
The hybrid estimator, parameterized by $w \in [0,1]$, is:
$$\tilde{\mu}_{H}(w) = w\tilde{\mu}_T + (1-w)\tilde{\mu}_L.$$
\end{defn}

\begin{lem} \label{lem:hybrid-mse}
$\tilde{\mu}_{H}(w)$ has expected squared error:
\begin{equation*}
\small
\begin{split}
\mathcal{E}_{H}(w) = &\frac{(w-c)^2}{cn}\sigma_T^2 + \frac{(w-c)^2}{(1-c)n}\sigma_L^2 + w^2 s_T^2 + \frac{(1-w)^2}{(1-c)n}s_L^2 \\
&+ (w\mu_T + (1-w)\mu_L - \mu)^2.
\end{split}
\end{equation*}
\end{lem}
\begin{proof}\renewcommand{\qedsymbol}{}
See Appendix~\ref{proof:hybrid-mse}.
\end{proof}
This estimator has all three types of error -- excess sampling error, privacy error, and bias error -- where the amounts of each error type depend on the weighting $w$.

\section{Homogeneous, Known-Variance Setting} \label{sec:homo-kv}

In this section, we design a hybrid estimator in the homogeneous setting which outperforms the baselines by carefully choosing a particular weighting for the hybrid estimator family from Definition~\ref{def:hybrid}.
To choose such a weighting, we restrict our focus to the homogeneous setting, where both groups' means are the same ($\mu = \mu_T = \mu_L$) and variances are the same ($\sigma^2 = \sigma^2_T = \sigma^2_L$).
 Beyond simplifying the expressions we're analyzing, the homogeneous setting eliminates bias error from our defined estimators, which removes any dependence on $\mu$ from the derived error expressions.
This is important, since the curator's goal is to learn $\mu$ from the data; thus, no particular knowledge of $\mu$ is assumed.
Therefore, in the homogeneous setting, a weighting can be chosen by analyzing the hybrid estimator's derived error expressions without needing any knowledge of $\mu$.
However, there is still excess sampling error for the estimators in this setting -- in other words, error expressions still depend on the data variance $\sigma^2$.
Thus, in this section, we make the common assumption in statistical literature that $\sigma^2$ is known to the curator, and derive and analyze the optimal hybrid estimator from the convex family.

\paragraph*{KVH Estimator} \label{sec:homo-kvh}

We now derive and analyze the ``known-variance hybrid'' (KVH) estimator by computing the optimal weighting $w^*$ that minimizes $\mathcal{E}_{H}(w)$.
This can be analytically computed and directly implemented by the curator, since each term of $\mathcal{E}_{H}(w)$ is known in this setting.
\begin{defn} \label{def:kvh}
The known-variance hybrid estimator in the homogeneous setting is:
$$\tilde{\mu}_{KVH} = w^* \tilde{\mu}_T + (1-w^*)\tilde{\mu}_L,$$
where $w^* = \frac{c(\sigma^2 + s_L^2)}{\sigma^2 + c(n s_T^2 (1-c) + s_L^2)}$ is obtained by minimizing $\mathcal{E}_{H}(w)$ with respect to $w$.
\end{defn}

\begin{lem} \label{lem:kvh-mse}
$\tilde{\mu}_{KVH}$ has expected squared error:
$$\mathcal{E}_{KVH} = \frac{(w^*-c)^2}{c(1-c)n}\sigma^2 + w^{*2} s_T^2 + (1-w^{*})^2\frac{s_L^2}{(1-c)n}.$$
\end{lem}
\noindent Although all users' data is used here, weighting the estimates by $w^*$ induces excess sampling error $\frac{(w^*-c)^2}{c(1-c)n}\sigma^2$, and the privacy error $w^{*2} s_T^2 + \frac{(1-w^*)^2}{(1-c)n}s_L^2$ is the weighted combination of the groups' privacy errors.\\

Now we compute and analyze the relative improvement in MSE of the KVH estimator over the best MSE of either the TCM-Only estimator or the Full-LM estimator.
\begin{thm}
The relative improvement of $\tilde{\mu}_{KVH}$ over the better of $\tilde{\mu}_T$ and $\tilde{\mu}_F$ is:
\begin{equation*}
R(\mathcal{E}_{KVH}) = \gamma \cdot
\begin{cases}
	\frac{1-c}{cn}\sigma^2 + s_T^2 & \text{if $c>c_{crit} \wedge n \le n_{crit}$} \\
	\frac{s_L^2}{n} & \text{otherwise}
  \end{cases},
\end{equation*}
where $\gamma = \frac{(1-c)\sigma^2 s_L^2 + cn(c\sigma^2+s_L^2)s_T^2}{n(\sigma^2 + cs_L^2)+(1-c)cn^2s_T^2}$ and $c_{crit}$ and $n_{crit}$ are as defined in Lemma~\ref{lem:crit}.
\end{thm}
\begin{proof}
Direct application of Lemmas~\ref{lem:tcm-only-mse}, \ref{lem:full-lm-mse}, and~\ref{lem:kvh-mse} to Definition~\ref{def:R-crit}.
\end{proof}

Algebraic analysis of this relative improvement reveals that $R(\mathcal{E}_{KVH}) > 1$ when the number of TCM users is less than $s_L^2 / s_T^2$.
For the common DP mechanisms that apply $0$-mean additive noise, this is trivially satisfied.
For instance, when adding $\epsilon$-DP Laplace noise, $s_L^2 / s_T^2 = c^2 n^2 \ge cn = |T|$.
Moreover, although $R(\mathcal{E}_{KVH})$ is theoretically unbounded, using the $\epsilon$-DP Laplace mechanism in the high-privacy regime ($\epsilon \le 1$) enables a tight characterization of the maximum possible relative improvement.
\begin{cor} \label{cor:relaistic-improvement}
The maximum relative utility of $\tilde{\mu}_{KVH}$ when using the Laplace mechanism in the high-privacy regime is bounded as:
$$17/8 \le \max_{\substack{\epsilon \le 1 \\ c,n,m,\sigma}} R(\mathcal{E}_{KVH}) \le 16/7.$$
\end{cor}
\begin{proof}\renewcommand{\qedsymbol}{}
See Appendix~\ref{proof:realistic-improvement}.
\end{proof}

\paragraph*{Empirical Evaluation of $R(\mathcal{E}_{KVH})$} \label{sec:exploring-R-kvh}

\begin{figure*}[h]
\centering
\begin{tabular}{cc}
  \includegraphics[width=.4\linewidth]{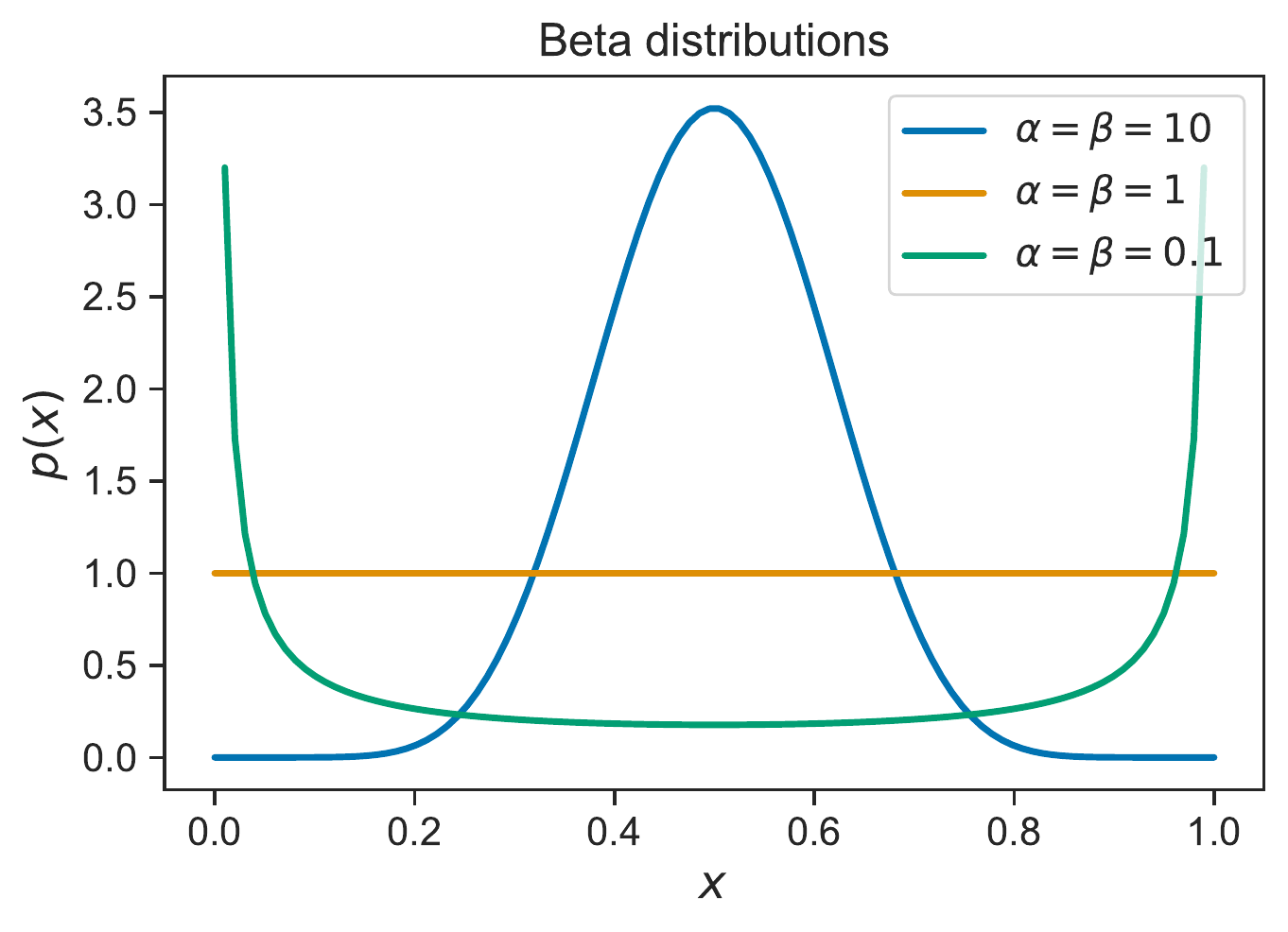} &  \includegraphics[width=.4\linewidth]{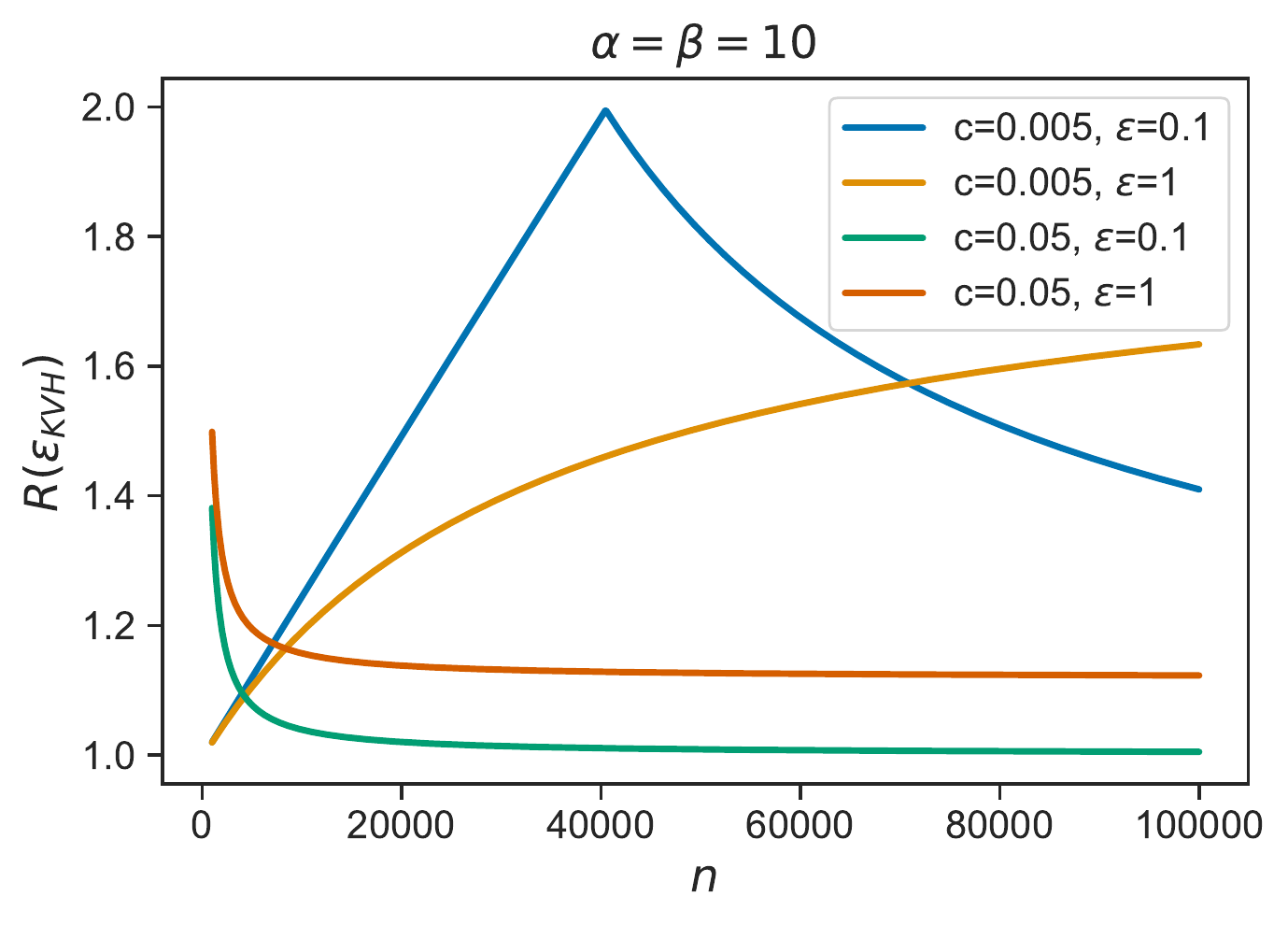} \\
  (a) & (b) \\
   \includegraphics[width=.4\linewidth]{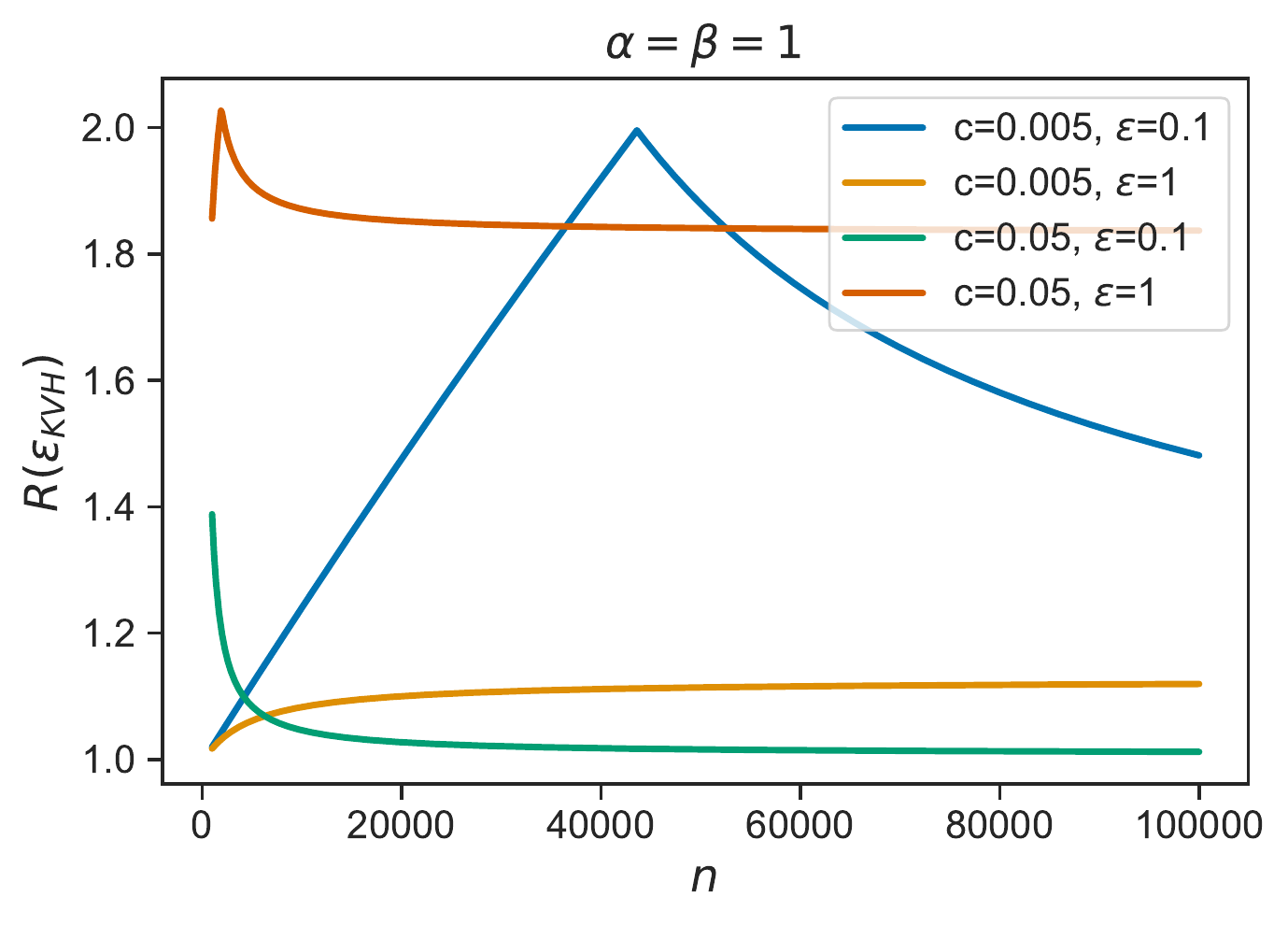} &   \includegraphics[width=.4\linewidth]{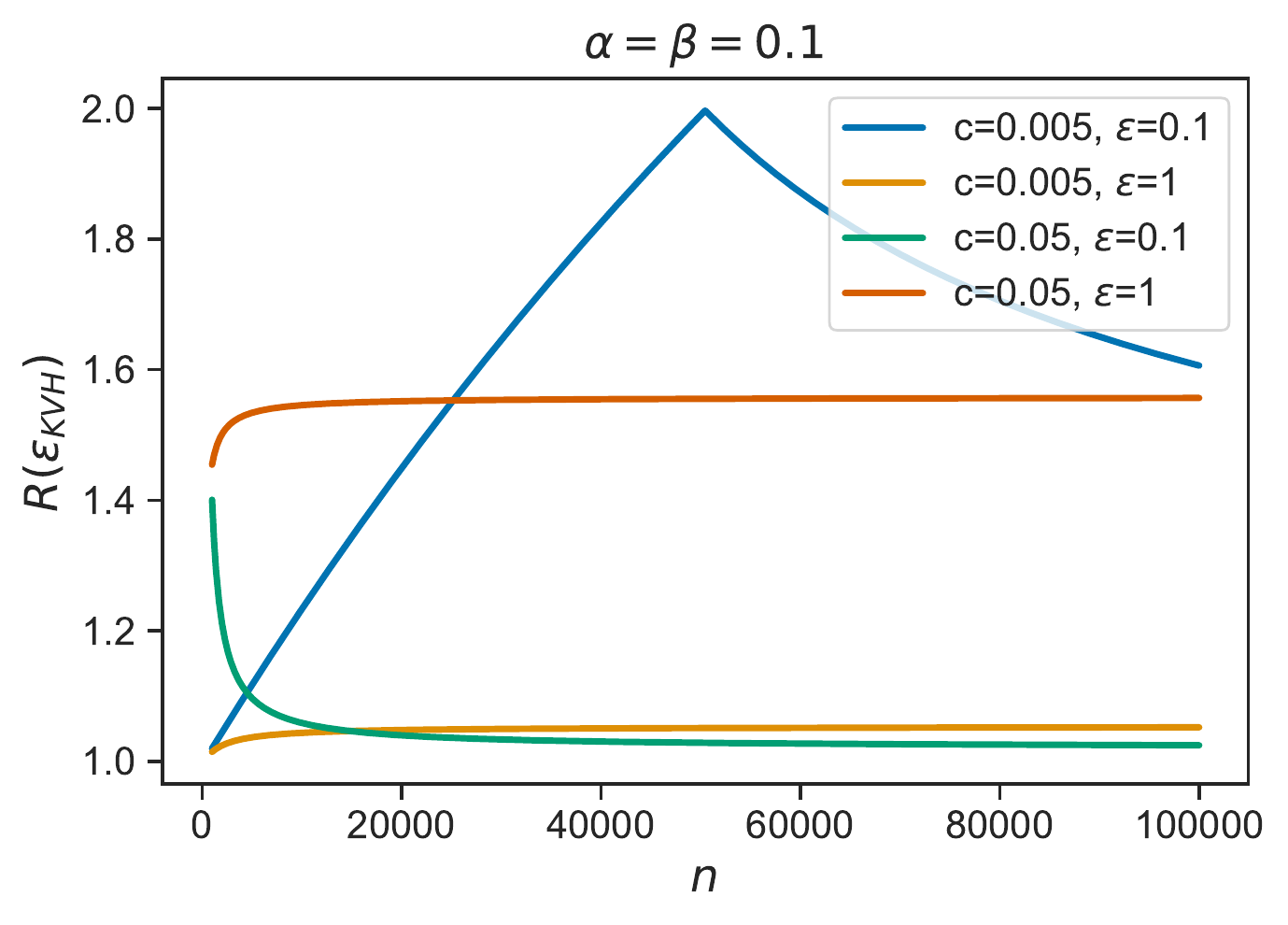}\\
(c) & (d) \\
\end{tabular}
\caption{(a) Probability density functions of Beta$(\alpha, \beta)$ distributions for various $\alpha, \beta$ values. (b,c,d) The relative improvement $R(\mathcal{E}_{KVH})$ for each Beta distribution across a range of $n$ values, for various $c$ and $\epsilon$ values.}
\label{fig:beta-R-kvh}
\end{figure*}

\begin{figure*}[h]
\centering
\begin{tabular}{cc}
 \hspace*{-1cm}
  \includegraphics[width=.4\linewidth]{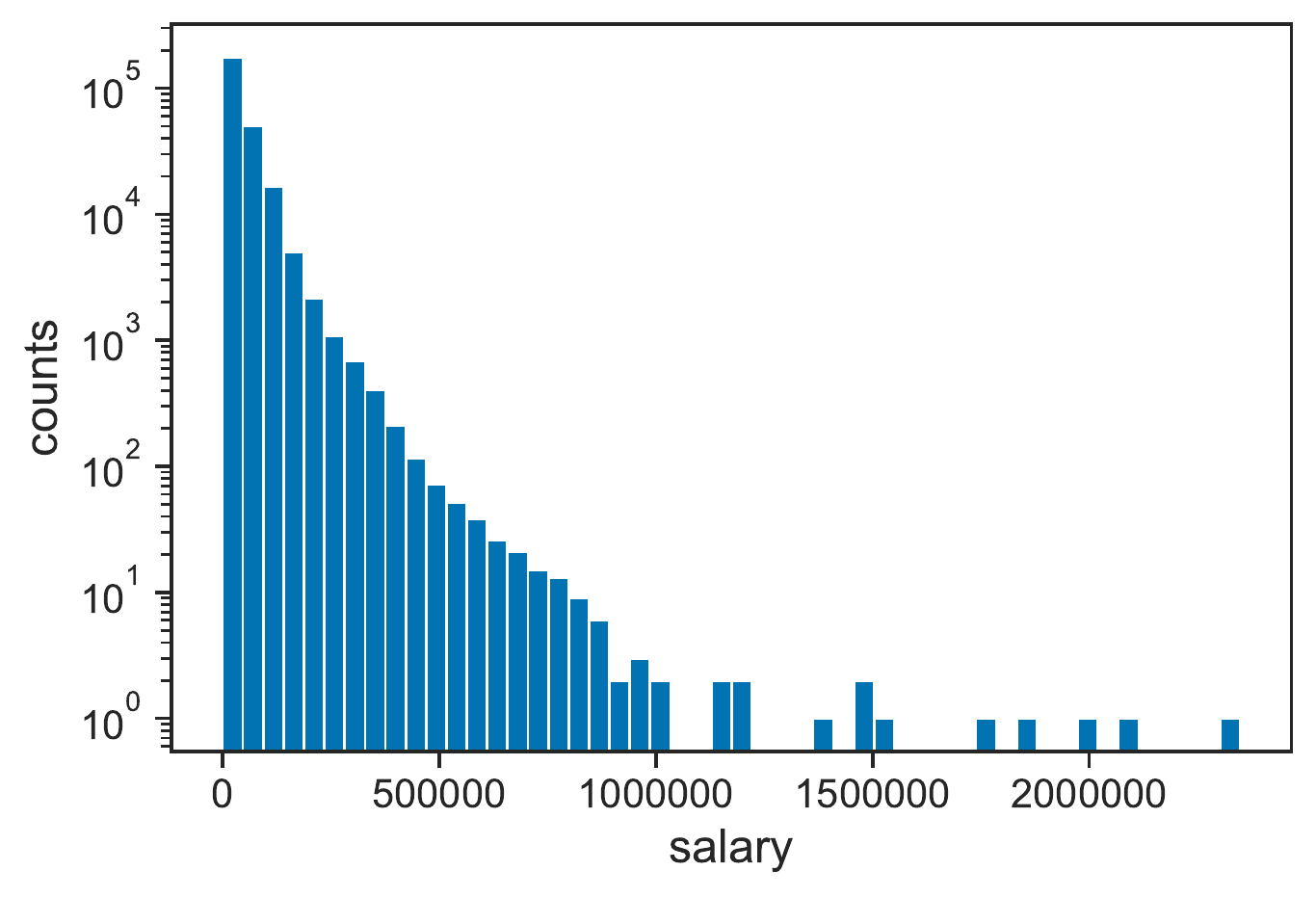} & \includegraphics[width=.4\linewidth]{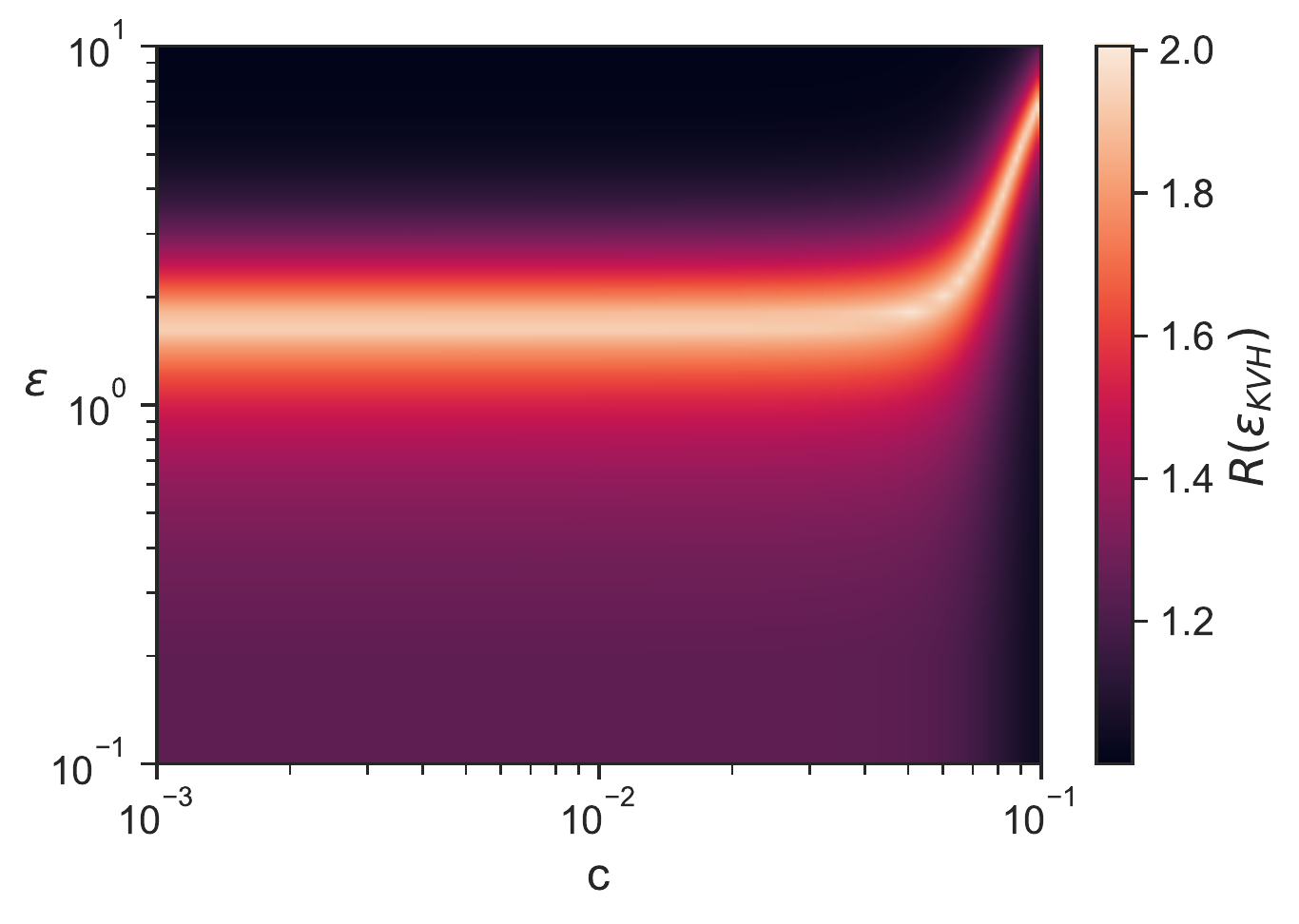} \\
  (a) & (b) \\
\end{tabular}
\caption{(a) Distribution of salaries of UC employees. (b) The relative improvement $R(\mathcal{E}_{KVH})$ across a range of $c$ and $\epsilon$ values.}
\label{fig:uc-salary-R-kvh}
\end{figure*}

To better understand what improvements one can expect from $\tilde{\mu}_{KVH}$ in practical applications, we empirically evaluate $R(\mathcal{E}_{KVH})$ using the $\epsilon$-DP Laplace mechanism in the context of various datasets.
Note that although the hybrid estimator's performance is dependent on the data distribution only through $\sigma$, $n$, and $m$, we use datasets to realistically motivate these values.

In Figure~\ref{fig:beta-R-kvh}, we use three synthetic datasets from the Beta$(\alpha, \beta)$ distribution: Beta$(10, 10)$, Beta$(1, 1)$, and Beta$(0.1, 0.1)$.
These symmetric distributions are chosen to induce different $\sigma$ values -- low ($\sigma \approx 0.109$), medium ($\sigma \approx 0.289$), and high ($\sigma \approx 0.456$).
For each distribution, $R(\mathcal{E}_{KVH})$ is plotted across $n \in [10^3, 10^5]$, $c \in \{0.5\%, 5\%\}$, and $\epsilon \in \{0.1, 1\}$.
Since the Beta distributions are supported on the interval $[0,1]$, we let $m=1$.
Figures~\ref{fig:beta-R-kvh}b,c,d show that in these settings, $R(\mathcal{E}_{KVH})$ is lower-bounded by $1$ and none are much larger than $2$ -- matching our theoretical analysis.
Observe that the ``peaking'' behavior of some curves is caused by the the $n_{crit}$ and $c_{crit}$ values being surpassed, which corresponds to the TCM group's data beginning to outperform the LM group's data in terms of MSE.
The curves which don't peak either have trivially surpassed the critical values (i.e., $n_{crit} < 1$ with $c > c_{crit}$) or have $c < c_{crit}$; importantly, they don't change behavior at some $n$ not shown in the figures.

In Figure~\ref{fig:uc-salary-R-kvh}, we use a real-world dataset of salaries of $n = 252,540$ employees in the University of California system in 2010 \cite{ucnet}.
This dataset was chosen due to its relatively high asymmetry, with a maximum salary of $m \approx 2,349,033$ and standard deviation of $\sigma \approx 53,254$ (both assumed to be known).
As $\sigma$, $n$, and $m$ are determined by the dataset, we examine the $R(\mathcal{E}_{KVH})$ values across a large space of the remaining free parameters: $c \in [0.1\%, 10\%]$ and $\epsilon \in [0.1, 10]$.
We see the relative improvement peak just above $2$ in the high-privacy regime, with this maximum improvement continuing into the low-privacy regime.

\section{Homogeneous, Unknown-Variance Setting} \label{sec:homo-uv}
In this section, we design a different hybrid estimator for the homogeneous setting, now applied to the case where the variance $\sigma^2$ of the data is not known.
This is a more realistic setting, as an analyst with no knowledge of the distribution's mean typically also doesn't have knowledge of its variance.

The KVH estimator was able to use knowledge of the variance to weigh the estimates of the two groups so that the trade-off of excess sampling error and privacy error was optimally balanced.
In this unknown-variance case, determining the optimal weighting is no longer viable.
Nevertheless, we can heuristically choose a weighting which may (or may not) perform well depending on the underlying distribution.
Thus, we propose a heuristic weighting choice for combining the groups' estimates and analyze it theoretically and empirically.
Before detailing this estimator, we first discuss a useful, but weaker, measure of relative improvement for this unknown-variance case.

\paragraph*{Utility Over At Least One Baseline}
Ideally, estimators would have $R(\mathcal{E}) \ge 1$ for all parameters.
If the regions can be computed where each baseline estimator has the best MSE, then a hybrid estimator can be designed to use this knowledge to trivially ensure $R(\mathcal{E}) \ge 1$.
However, depending on the setting (such as when variance is unknown), determining these regions precisely may not be feasible.
In these cases, we want to at least ensure that the hybrid estimator is never performing worse than both baselines, and do so by defining the following measure of relative improvement.
\begin{defn} \label{def:r}
The relative improvement of an estimator $\mathcal{E}$ over the worst baseline estimator is:
$$r(\mathcal{E}) = \frac{\max\{\mathcal{E}_T, \mathcal{E}_F\}}{\mathcal{E}}.$$
\end{defn}

Our characterization of the critical values in Lemma~\ref{lem:crit} enables $r(\mathcal{\varepsilon})$ to be re-written as follows.

\begin{defn} \label{def:r-crit}
The relative improvement of an estimator with MSE $\mathcal{E}$ over the worst baseline estimator is:
\begin{equation*}
  r(\mathcal{E}) = \frac{1}{\mathcal{E}} \cdot
  \begin{cases}
	\frac{s_L^2}{n} & \text{if $c>c_{crit} \wedge n \le n_{crit}$} \\
	\frac{1-c}{cn}\sigma^2 + s_T^2 & \text{otherwise}
  \end{cases}
\end{equation*}
\end{defn}

We remark that although any ``reasonable'' hybrid estimator should satisfy $r(\mathcal{E}) \ge 1$, this criteria is not automatically satisfied.
Even among the family of hybrid estimators from Definition~\ref{def:hybrid}, there exist estimators which have $r(\mathcal{E}_H(w)) < 1$ in some regimes.
Concretely, consider an arbitrary constant as the weight; e.g., $w=0.001$.
Using the parameters from experiments ($m=1, \sigma=1/6, c=0.01, \epsilon=0.1$), we have $r(\mathcal{E}_H(0.001)) < 1$ for $n \geq 10,058$.
Thus, estimators must be designed carefully to maximize utility and, at the very least, ensure $r(\mathcal{E}) > 1$ everywhere.

We now propose and analyze a hybrid estimator with a heuristically-chosen weighting that is based on the amount of privacy noise each group adds.
However, we first remark that we additionally investigated a na\"{\i}ve weighting heuristic, which performs the same weighting as the non-private benchmark estimator: weight the estimates based purely on the group size (i.e., $w=c$).
Our empirical evaluations showed that for practical parameters, this estimator is typically inferior to the estimator we're about to discuss.
Thus, we have omitted it from this presentation for brevity. 

\paragraph*{PWH Estimator}
We choose this heuristic weighting by considering only the induced privacy error of each groups' estimate.
Thus, we refer to this as the ``privacy-weighted hybrid'' (PWH) estimator.
Note that this weighting seeks solely to optimally balance privacy error between the groups, and therefore ignores the induced excess sampling error.
Explicitly, from $\mathcal{E}_{H}(w)$
of Lemma~\ref{lem:hybrid-mse} applied to the homogeneous setting, this weighting corresponds to choosing $w$ to minimize $w^2 s_T^2 + (1-w)^2\frac{s_L^2}{(1-c)n}$, stated in the following definition.

\begin{defn} \label{def:pwh}
The privacy-weighted hybrid estimator is:
$$\tilde{\mu}_{PWH} = w_{PWH}\tilde{\mu}_T + (1-w_{PWH})\tilde{\mu}_L,$$
where $w_{PWH} = \frac{s_L^2}{s_L^2 + (1-c)ns_T^2}$
\end{defn}

\begin{lem} \label{lem:pwh-mse}
$\tilde{\mu}_{PWH}$ has expected squared error:
${\scriptstyle \mathcal{E}_{PWH} = \frac{(1-c) c n^2 s_T^4 \left(c \sigma ^2+s_L^2\right)+c n s_L^2 s_T^2 \left(2 (c-1) \sigma ^2+s_L^2\right)+(1-c) \sigma ^2 s_L^4}{c n \left(s_L^2+(1-c) n s_T^2\right){}^2}.}$\\
\end{lem}

\noindent This estimator has a mixture of both excess sampling error and privacy error.
Since the privacy error was directly optimized, we expect this estimator to do well when the data variance $\sigma^2$ is small, as this will naturally induce small excess sampling error.

Now we are able to discuss the relative improvement of the PWH estimator over the baselines.

\begin{thm}
The relative improvements of the PWH estimator $\tilde{\mu}_{PWH}$ over $\tilde{\mu}_T$ and $\tilde{\mu}_F$ are:
\begin{equation*}
R(\mathcal{E}_{PWH}) = \gamma \cdot
\begin{cases}
	\frac{1-c}{cn}\sigma^2 + s_T^2 & \text{if $c>c_{crit} \wedge n \le n_{crit}$} \\
	\frac{s_L^2}{n} & \text{otherwise}
  \end{cases},
\end{equation*}
\begin{equation*}
r(\mathcal{E}_{PWH}) = \gamma \cdot
\begin{cases}
	\frac{s_L^2}{n} & \text{if $c>c_{crit} \wedge n \le n_{crit}$} \\
	\frac{1-c}{cn}\sigma^2 + s_T^2 & \text{otherwise}
  \end{cases},
\end{equation*}
where ${\scriptstyle \gamma = \frac{c n \left(s_L^2+(1-c) n s_T^2\right){}^2}{(1-c) c n^2 s_T^4 \left(c \sigma ^2+s_L^2\right)+c n s_L^2 s_T^2 \left(2 (c-1) \sigma ^2+s_L^2\right)+(1-c) \sigma ^2 s_L^4}}$ and $c_{crit}$ and $n_{crit}$ are as defined in Definition~\ref{def:R-crit}.
\end{thm}
\begin{proof}
Direct application of Lemmas~\ref{lem:tcm-only-mse}, \ref{lem:full-lm-mse}, and~\ref{lem:pwh-mse} to: Definition~\ref{def:R-crit} for $R(\mathcal{E}_{PWH})$, and Definition~\ref{def:r-crit} for $r(\mathcal{E}_{PWH})$.
\end{proof}

With the generic noise-addition privacy mechanisms, algebraic analysis of the weaker relative improvement measure reveals $r(\mathcal{E}_{PWH}) > 1$ unconditionally.
However, the regions where $R(\mathcal{E}_{PWH})$ is greater than $1$ are difficult to obtain analytically with these generic mechansisms.
By restricting our attention to the Laplace mechanism, we find that $R(\mathcal{E}_{PWH}) > 1$ is satisfied under certain conditions.
The first is a ``low relative privacy'' regime where $\epsilon \ge \frac{\sqrt{2}m}{\sigma}$; that is, once $\epsilon$ is large enough, we have $R(\mathcal{E}_{PWH}) > 1$.
For $\epsilon$ under this threshold, achieving $R(\mathcal{E}_{PWH}) > 1$ requires the following conditions on $c$ and $n$: either $c \le \frac{\epsilon^2 \sigma^2}{2m^2}$, or $c > \frac{\epsilon^2 \sigma^2}{2m^2} \ \wedge\ n < \frac{2m^2(1+c)}{c(2cm^2 - \epsilon^2 \sigma^2)}$.

\paragraph*{Empirical Evaluation of $R(\mathcal{E}_{PWH})$ and $r(\mathcal{E}_{PWH})$} \label{sec:exploring-pwh}

Here, we perform an empirical evaluation analogous to that done in Section~\ref{sec:exploring-R-kvh}.
Figure~\ref{fig:beta-pwh} presents $R(\mathcal{E}_{PWH})$ (top row) and $r(\mathcal{E}_{PWH})$ (bottom row) using the same Beta distributions and parameters ($n \in [10^3, 10^5]$, $c \in \{0.5\%, 5\%\}$, and $\epsilon \in \{0.1, 1\}$).
We find that there are many regions where $R(\mathcal{E}_{PWH})$ achieves a value of just greater than $1$, and some regions where it achieves values competitive with the KVH estimator.
Unsurprisingly, since this weighting is chosen without accounting for the variance, there are also clear regions where the $R(\mathcal{E}_{PWH})$ is noticeably less than $1$.
Even in the regions where $R(\mathcal{E}_{PWH})$ is low, the $r(\mathcal{E}_{PWH})$ values in the bottom row often show that the PWH estimator significantly improves over the worse of the two baseline estimators.
An empirical evaluation of this estimator on the UC salaries dataset can be found in Appendix~\ref{app:pwh-uc-salary}.

\begin{figure*}[h]
\centering
\begin{tabular}{ccc}
  \includegraphics[width=.33\linewidth]{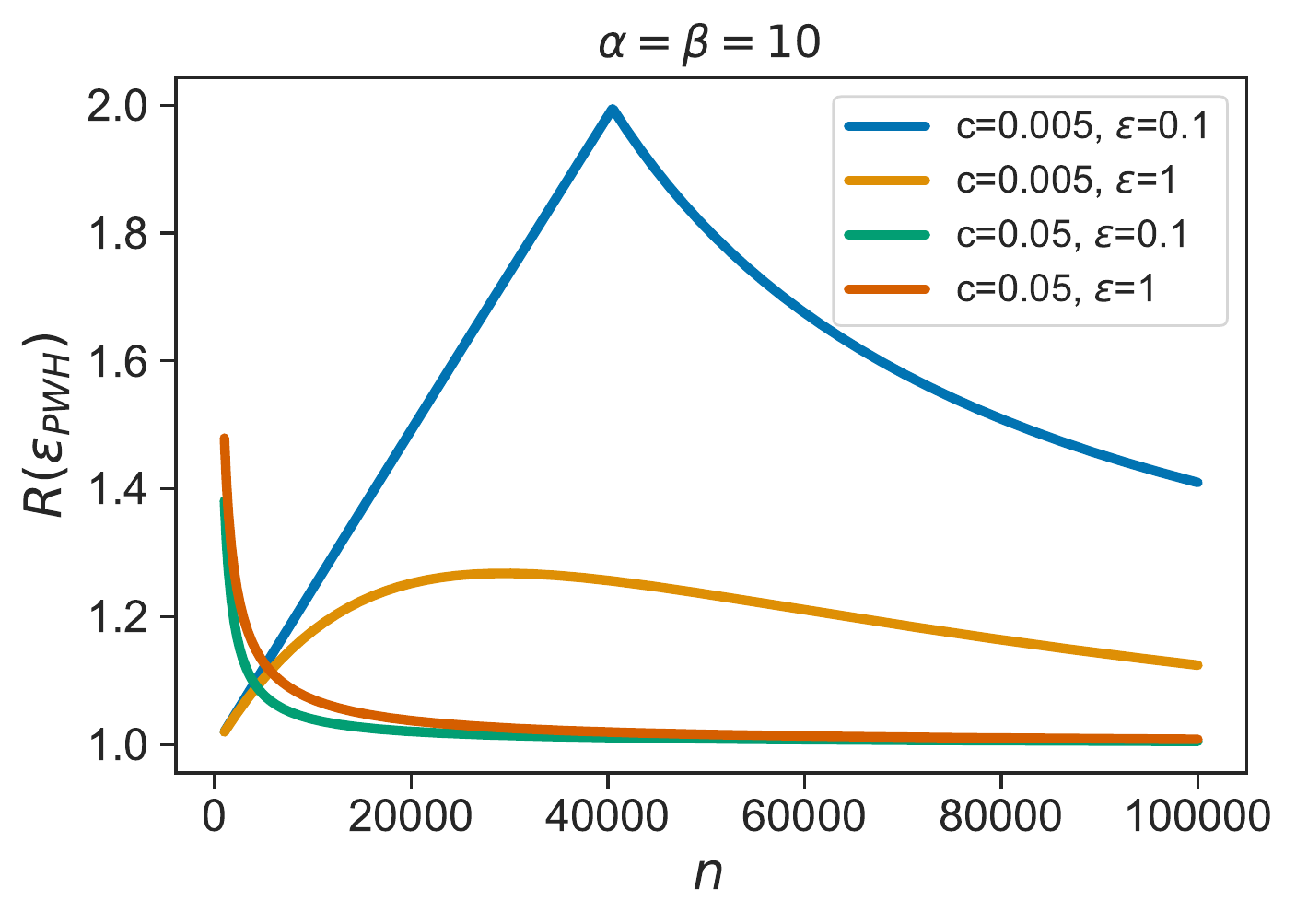} &  \includegraphics[width=.33\linewidth]{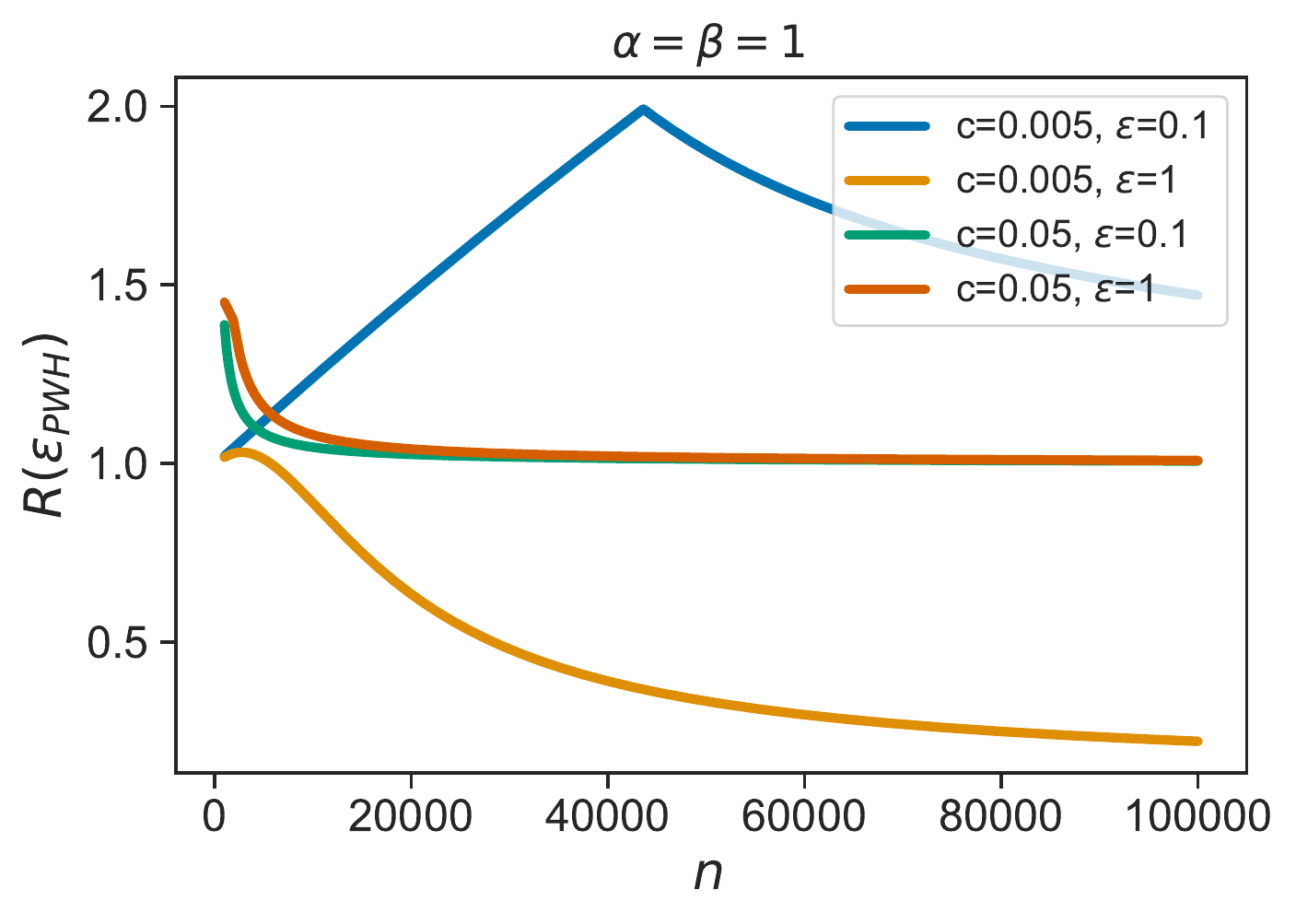} &
  \includegraphics[width=.33\linewidth]{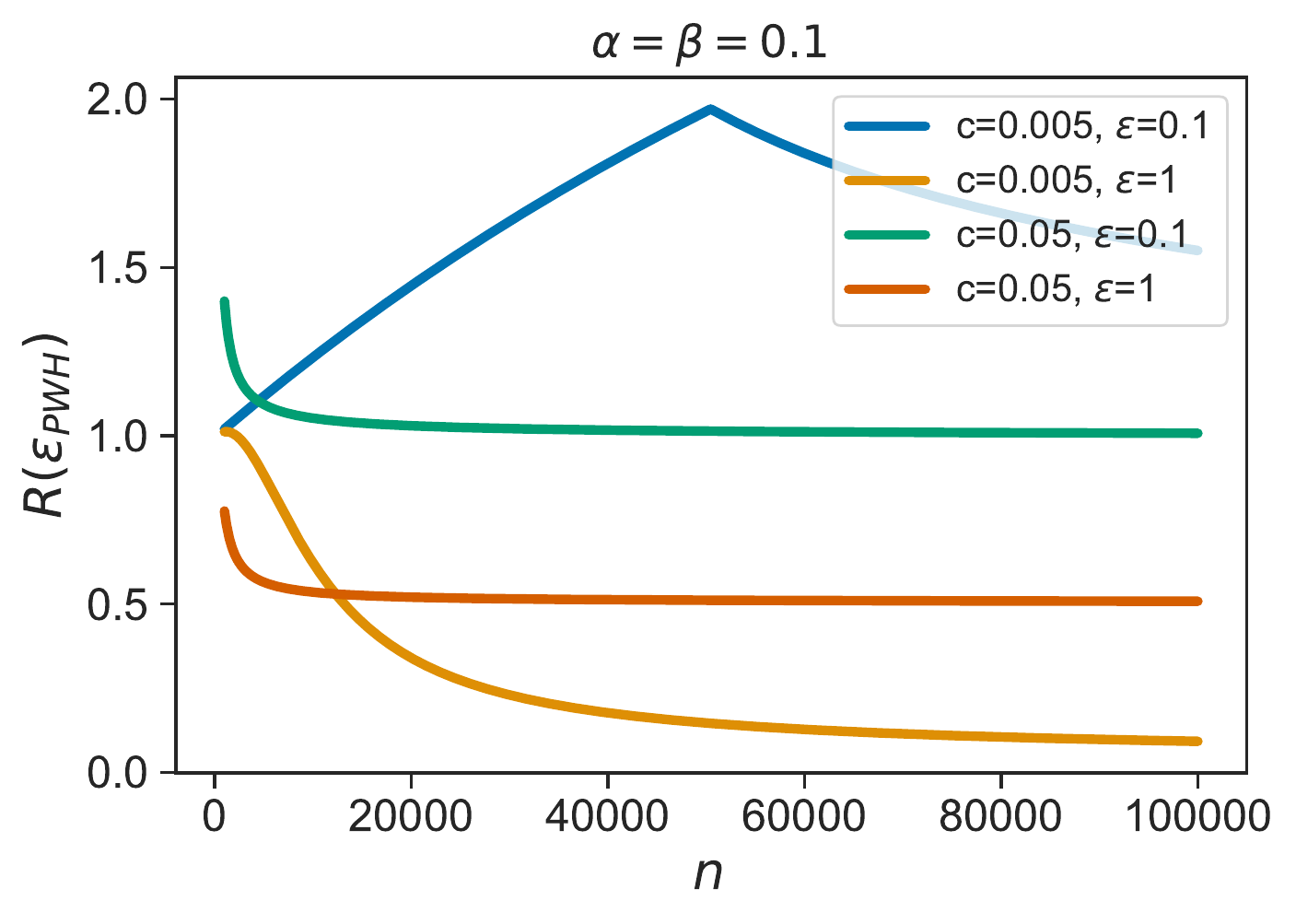}\\
(a) & (b) & (c) \\
 \includegraphics[width=.33\linewidth]{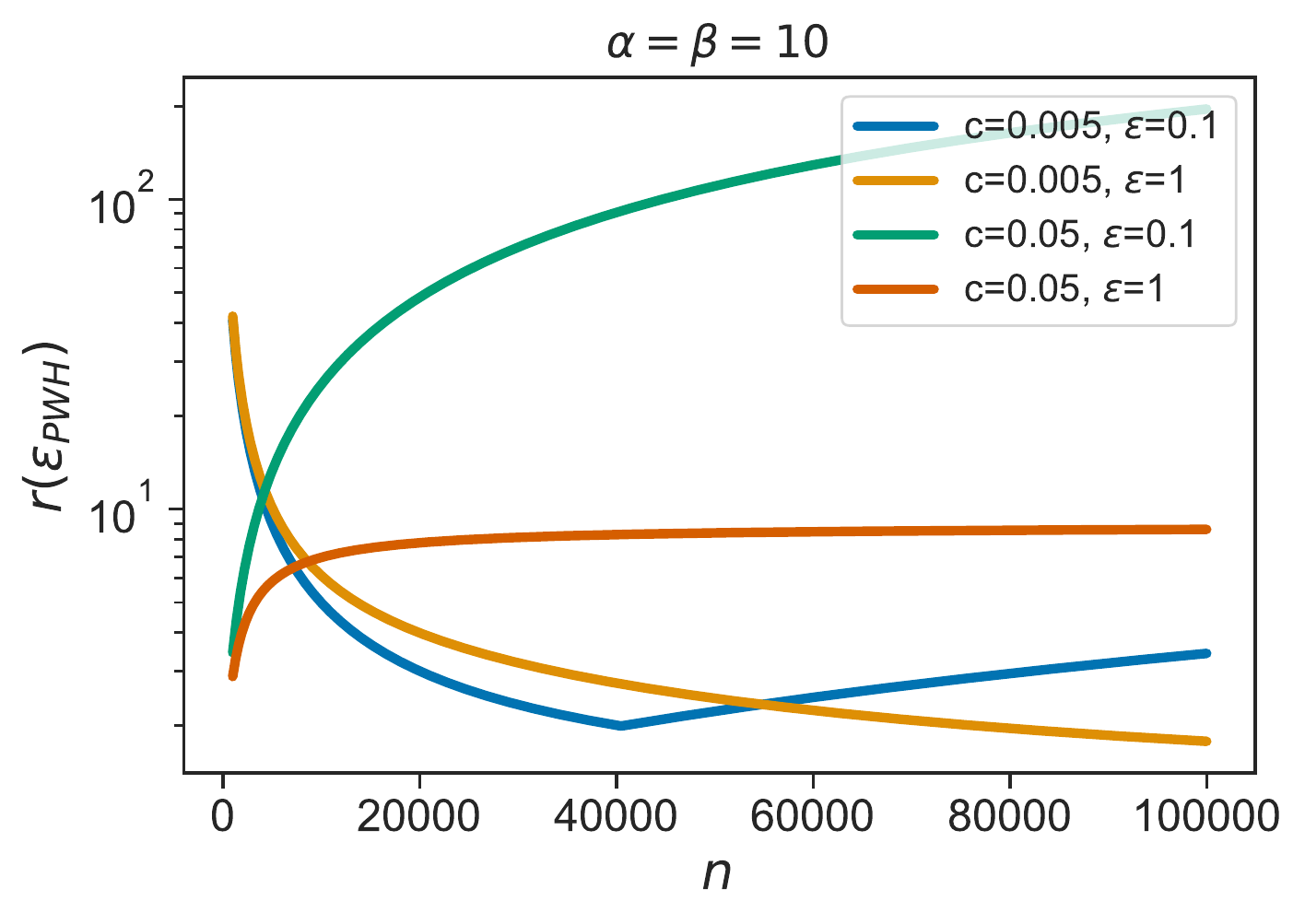} &   \includegraphics[width=.33\linewidth]{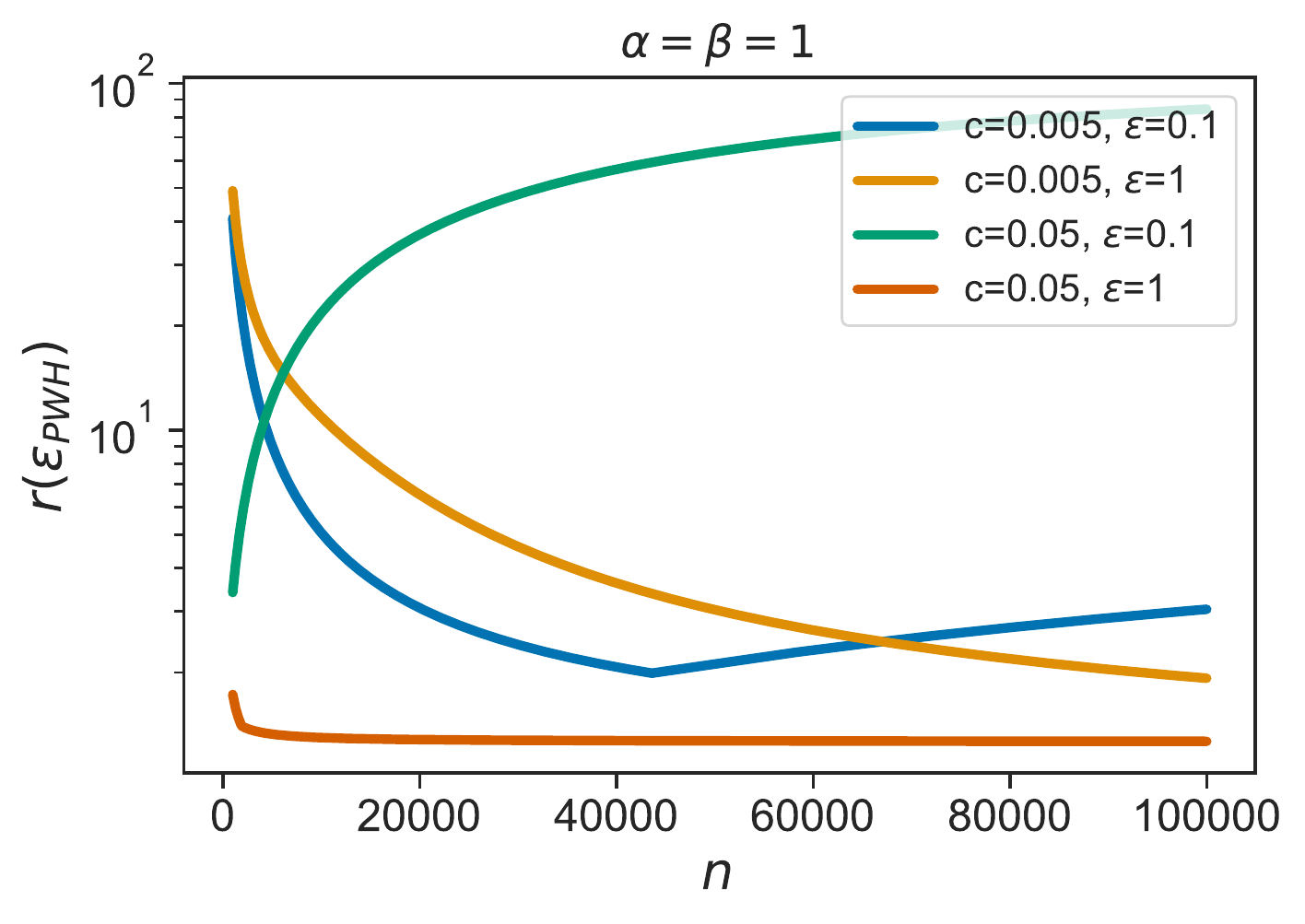} &
 \includegraphics[width=.33\linewidth]{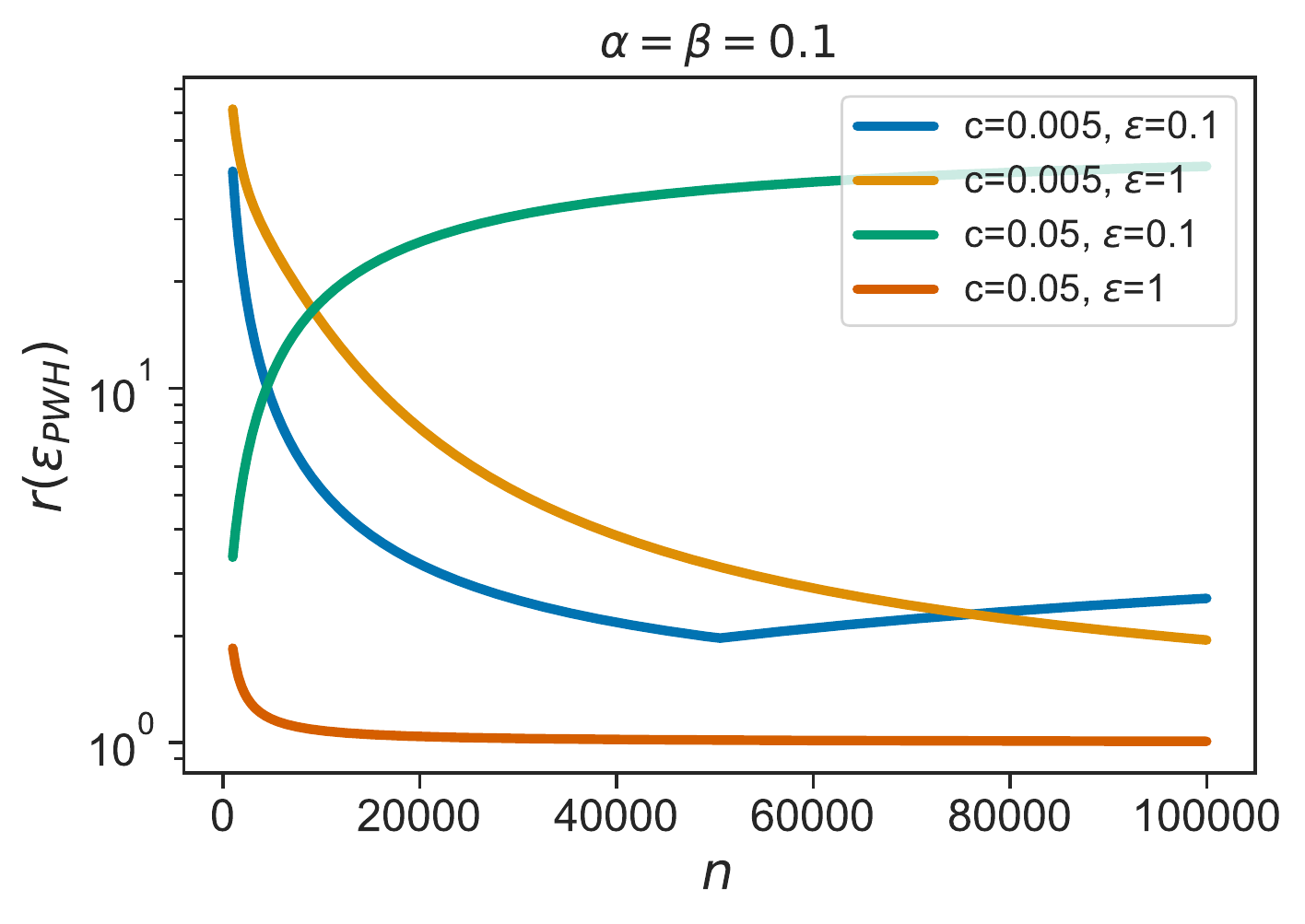} \\
(d) & (e) & (f) \\
\end{tabular}
\caption{Across a range of $n$ values, for various $c$ and $\epsilon$ values for each Beta distribution (plotted in Figure~\ref{fig:beta-R-kvh}a): (a,b,c) shows $R(\mathcal{E}_{PWH})$ values and (d,e,f) shows $r(\mathcal{E}_{PWH})$ values.}
\label{fig:beta-pwh}
\end{figure*}

\section{Heterogeneous Setting} \label{sec:hetero}
In this section, we examine the effects of the groups' distributions diverging on the quality of our estimators.
This is motivated by the fact that the hybrid model allows users to self-partition based on their trust preferences.
Such self-partitioning may cause the groups' distributions to be different.
For instance, since the TCM users have similar trust preferences, their data may also be more similar than the LM users'.
This could manifest as variance-skewness between the groups.
Alternatively, the TCM users may have fundamentally different data than the LM users, which would manifest as mean-skewness between the groups.
Thus, we examine the case where the group means are the same but their variances are different, as well as the case where the group means are different but their variances are the same.
To understand these skewness effects, we empirically evaluate $R(\mathcal{E}_{KVH})$\footnote{We also performed the same empirical evaluation with the unknown-variance PWH estimator. The results were nearly identical to the KVH estimator's, and the conclusions were the same. Thus, we omitted them for brevity.}.

Although the heterogeneous setting is more general and complex, we can still derive the optimal weighting for the KVH estimator analogously to homogeneous KVH weighting of Definition~\ref{def:kvh}.
\begin{defn} \label{def:kvh-hetero}
The known-variance hybrid estimator in the heterogeneous setting is:
$$\tilde{\mu}_{KVH} = w^* \tilde{\mu}_T + (1-w^*)\tilde{\mu}_L,$$
where $w^* = \frac{c(s_L^2 + c\sigma_L^2+(1-c)(n(\mu_L - \mu)(\mu_L-\mu_T)+\sigma_T^2))}{cs_L^2+(1-c)cn((\mu_L-\mu_T)^2+s_T^2)+c\sigma_L^2+(1-c)\sigma_T^2}$
\end{defn}

\paragraph*{Variance-Skewness}
Here, we examine the case where $\mu_T = \mu_L$ but $\sigma_T^2 \neq \sigma_L^2$.
This reduces the KVH estimator's weighting to $w^* = \frac{c(s_L^2 + c\sigma_L^2+(1-c)\sigma_T^2)}{cs_L^2+(1-c)cns_T^2+c\sigma_L^2+(1-c)\sigma_T^2}$.
To gain insight into the effect of variance-skewness, we recall two Beta distributions previously used in our empirical evaluations: the low-variance Beta$(10,10)$ distribution ($\sigma = 0.109$) and the high-variance Beta$(0.1, 0.1)$ distribution ($\sigma = 0.456$).
We evaluate $R(\mathcal{E}_{KVH})$ in two scenarios: when the TCM group has data drawn from the low-variance distribution but the LM group has data drawn from the high-variance distribution, and vice versa.
Figure~\ref{fig:beta-R-hetero-variances} gives the results across the same range of $n$, $c$, and $\epsilon$ values as used in previous experiments.

The similarities between Figure~\ref{fig:beta-R-hetero-variances} and Figure~\ref{fig:beta-R-kvh} demonstrate that our estimator is robust to deviations in the LM group's variance.
For example, Figure~\ref{fig:beta-R-kvh}b shows $R(\mathcal{E}_{KVH})$ when all the data is from the low-variance distribution; that figure nearly exactly matches Figure~\ref{fig:beta-R-hetero-variances}a despite the fact that most of the data is now from the LM group's high-variance distribution.
As this applies to both of Figure~\ref{fig:beta-R-kvh}'s graphs, it is clear that the relative improvement heavily depends on the variance of the TCM group, regardless of whether the LM group had the low- or high-variance data.
In fact, in both graphs, the difference in relative improvement from the homogeneous case with variance $\sigma^2$ to the heterogeneous case where only the TCM group has variance $\sigma_T^2 = \sigma^2$ does not vary by more than $\pm 0.1$, and, typically, varies by less than $\pm 0.01$.

\begin{figure*}[h]
\centering
\setlength\extrarowheight{-10pt}
\begin{tabular}{cc}
  \includegraphics[width=.4\linewidth]{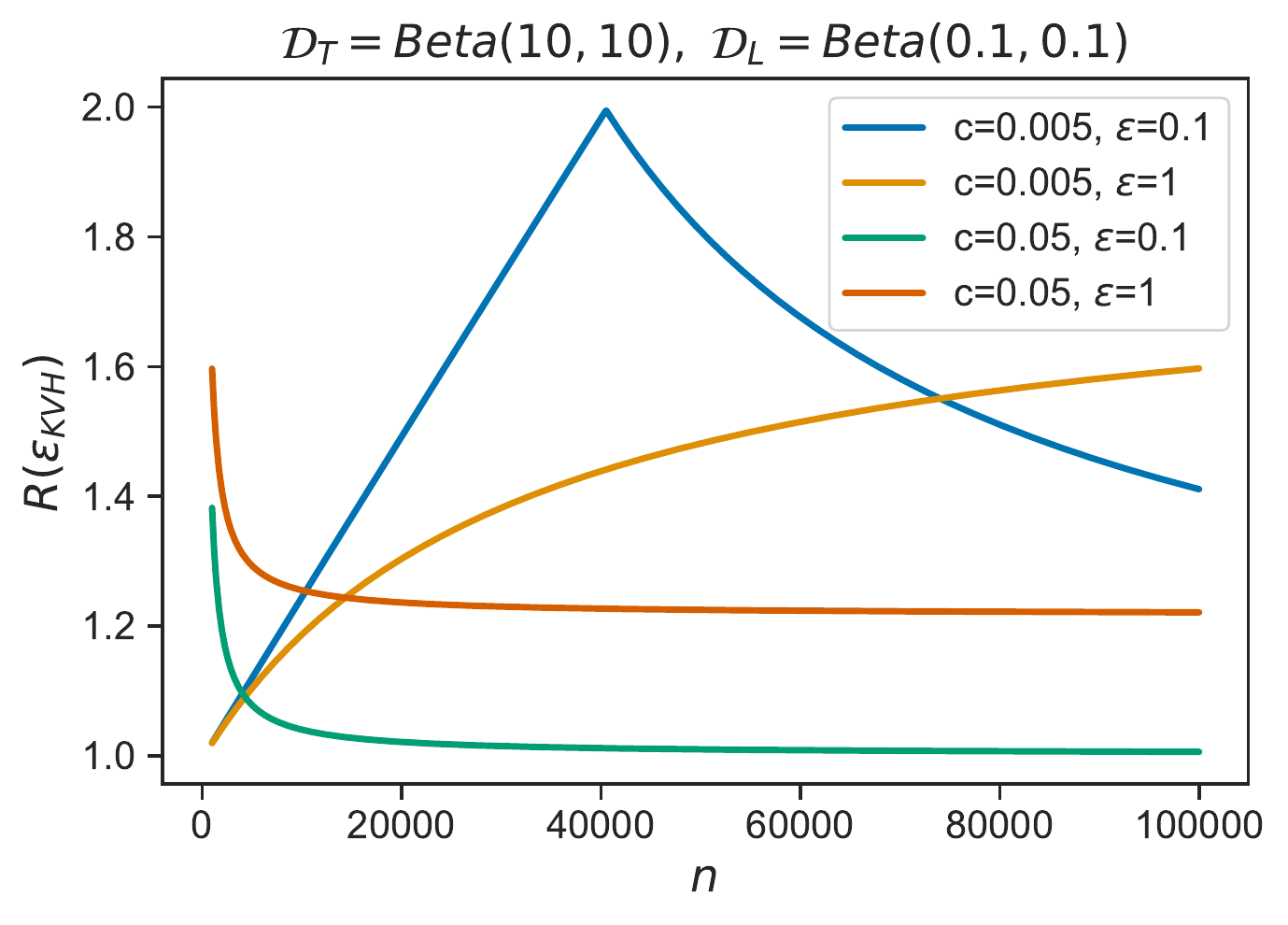} & \includegraphics[width=.4\linewidth]{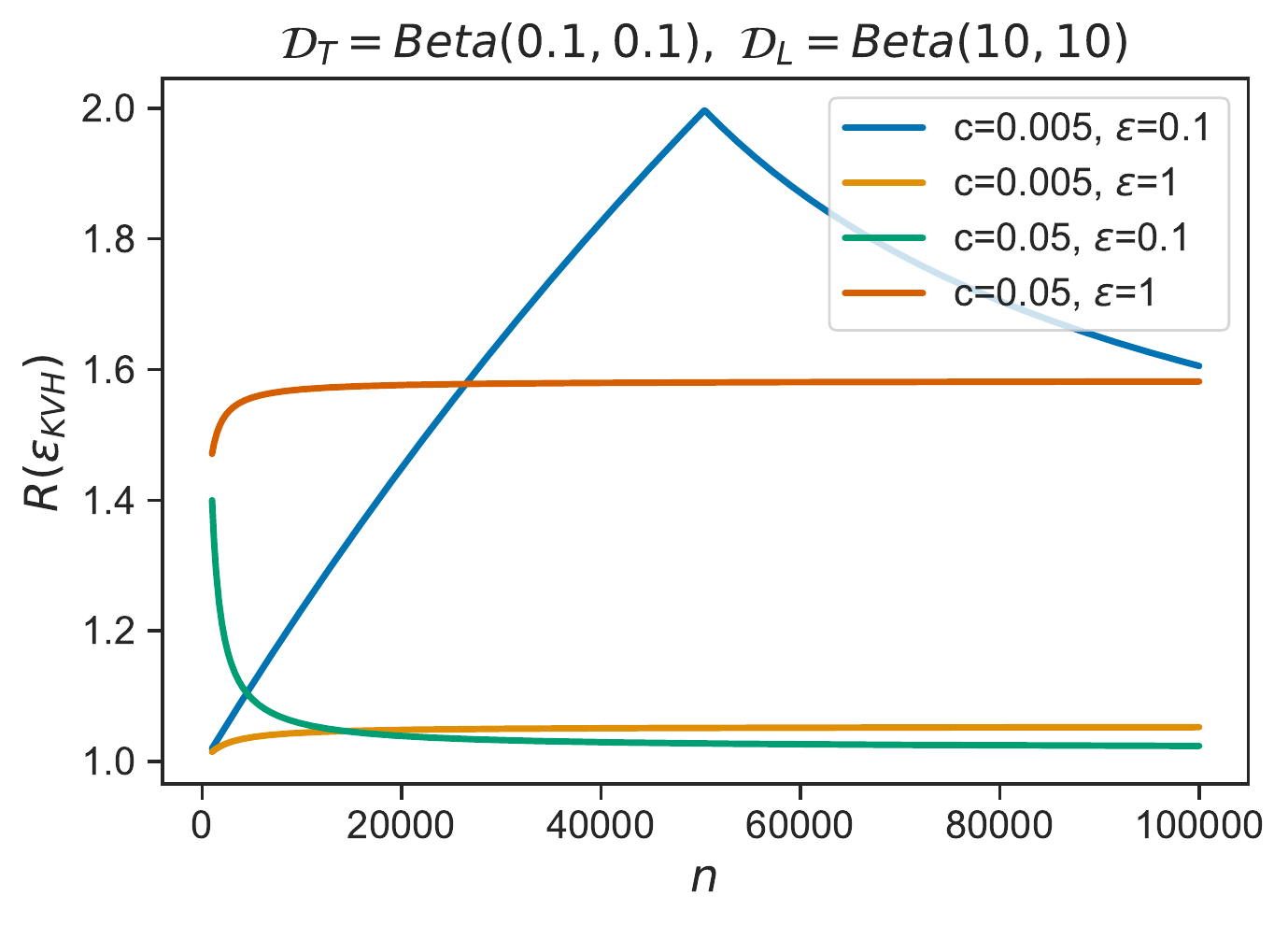}\\
(a) & (b) \\
\end{tabular}
\caption{The relative improvement $R(\mathcal{E}_{KVH})$ values when: (a) the TCM group has low-variance data but the LM group has high, and (b) when the TCM group has high-variance data but the LM group has low.}
\label{fig:beta-R-hetero-variances}
\end{figure*}

\paragraph*{Mean-Skewness}
Here, we examine the case where $\mu_T \neq \mu_L$ but $\sigma_T^2 = \sigma_L^2$.
This reduces the KVH estimator's weighting to $w^* = \frac{c(s_L^2 +(1-c)n(\mu_L - \mu)(\mu_L-\mu_T) + \sigma^2)}{cs_L^2+(1-c)cn((\mu_L-\mu_T)^2+s_T^2)+\sigma^2}$.
Importantly, this expression depends on the curator's knowledge of $\mu_T$ and $\mu_L$ -- an unreasonable requirement, since the curator's overarching goal is to learn the mean from the user data.
For applications where the groups' means are assumed to be different, computing separate estimates of each group's mean in their respective trust models would likely be more useful than a joint estimate.
Thus, we instead explore mean-skewness from the point of view of a curator who \emph{mistakenly believes} they are operating in the homogeneous setting, and thus uses the homogeneous weighting from Definition~\ref{def:kvh}.
This is useful in practice, as it demonstrates how a curator can use our analytical expressions for their specific problem instance to understand how utility is affected by misspecified assumptions about user data.

To analyze this case, we set up the following experiment, displayed in Figure~\ref{fig:beta-R-hetero-means}.
At the high level, we start with the control for the experiments: set both groups to the same distribution $\mathcal{D}_T = \mathcal{D}_L$ and obtain $R(\mathcal{E}_{KVH})$.
Next, we retain the distributional shape for both groups, but shift them in opposite directions; e.g., $\mathcal{D}_T - t, \mathcal{D}_L + t$ for some $t$.
We obtain the new $R(\mathcal{E}_{KVH})$ values under these distributions, and compare against the un-shifted results.
For clarity, we denote the relative improvement on the $t$-shifted distribution as $R^t(\mathcal{E}_{KVH})$.

We expect that as the divergence in means $t$ increases, the relative utility of our hybrid estimator will decrease.
To test this hypothesis concretely, we use the medium-variance Beta$(1,1)$ distribution ($\sigma = 0.289$) from our previous empirical evaluations as the experiment's base distribution.
We center this distribution at $1$ without rescaling, inducing support on $[0.5,1.5]$.
Then we set both $\mathcal{D}_T$ and $\mathcal{D}_L$ to this distribution, and obtain $R^0(\mathcal{E}_{KVH})$ on it (Figure~\ref{fig:beta-R-hetero-means}ab).
Next, we add a small shift of $t=0.25$ to each of the groups' distributions in opposite directions; i.e., $\mathcal{D}_T-0.25$ and $\mathcal{D}_L+0.25$, so that so that $|\mu_T - \mu_L|=0.5$.
These distributions, along with the corresponding $R^{0.25}(\mathcal{E}_{KVH})$ results, are shown in the second column of Figure~\ref{fig:beta-R-hetero-means}.
Finally, the third column of Figure~\ref{fig:beta-R-hetero-means} shows the analogous distributions and results when a large shift of $t=0.5$ is added so that $|\mu_T - \mu_L|=1$.\footnote{One caveat to these shifts is that as the data distribution becomes wider, the noise required to ensure DP must increase. Since we are interested in the effect of mean-skewness here, and not the effect of distribution-width, we conservatively fix $m=2$ for all experiments. That is, the same level of noise is used across shift-amounts, even if less noise may have sufficed to ensure DP.}
Unsurprisingly, these results depict a clear negative impact on the relative improvement as the means diverge, showing that our estimator is sensitive to skewness in the groups' means.

\begin{figure*}[h]
\centering
\setlength\extrarowheight{-10pt}
\begin{tabular}{c|cc}
  \includegraphics[width=.33\linewidth]{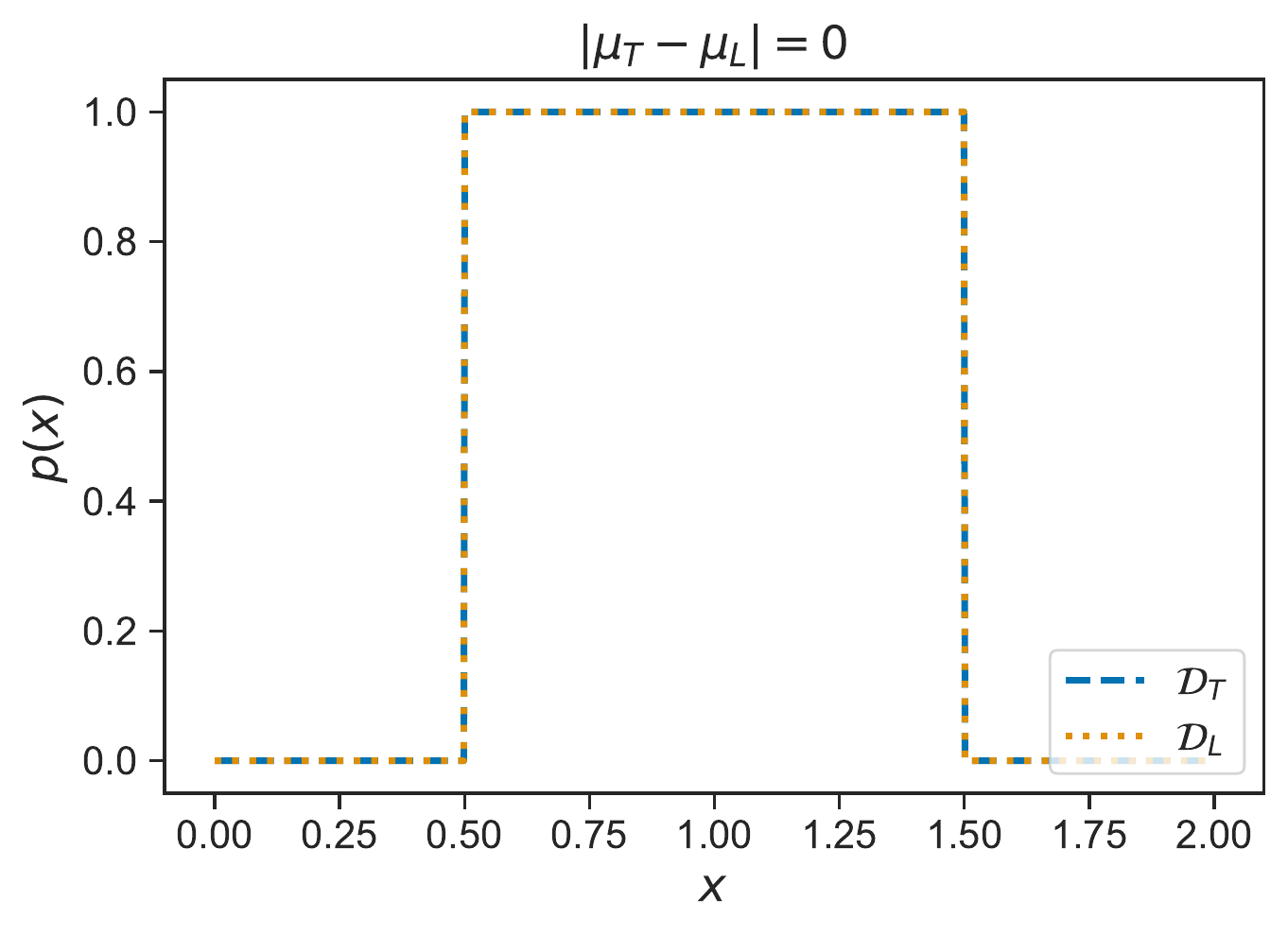} &  \includegraphics[width=.33\linewidth]{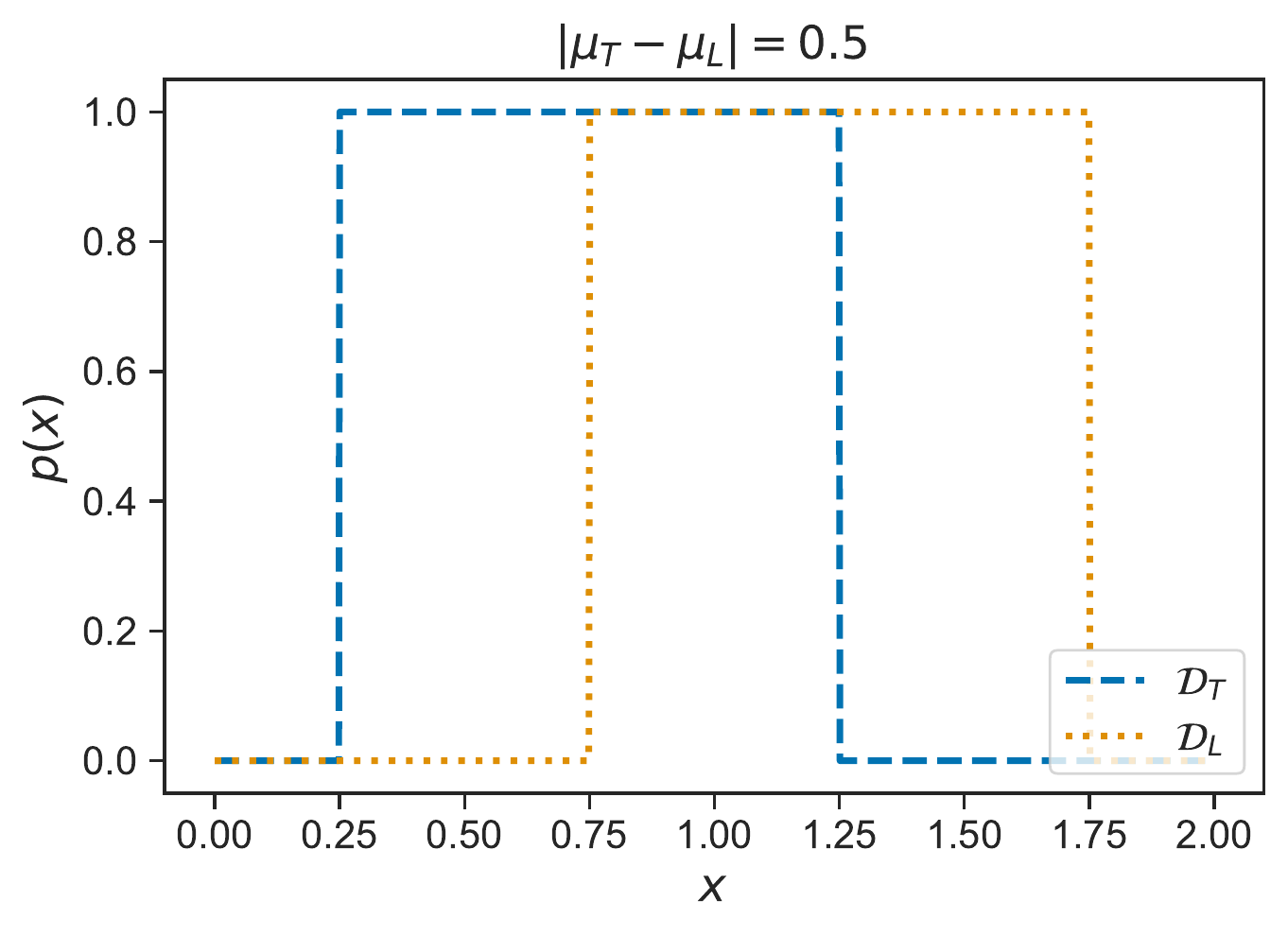} & \includegraphics[width=.33\linewidth]{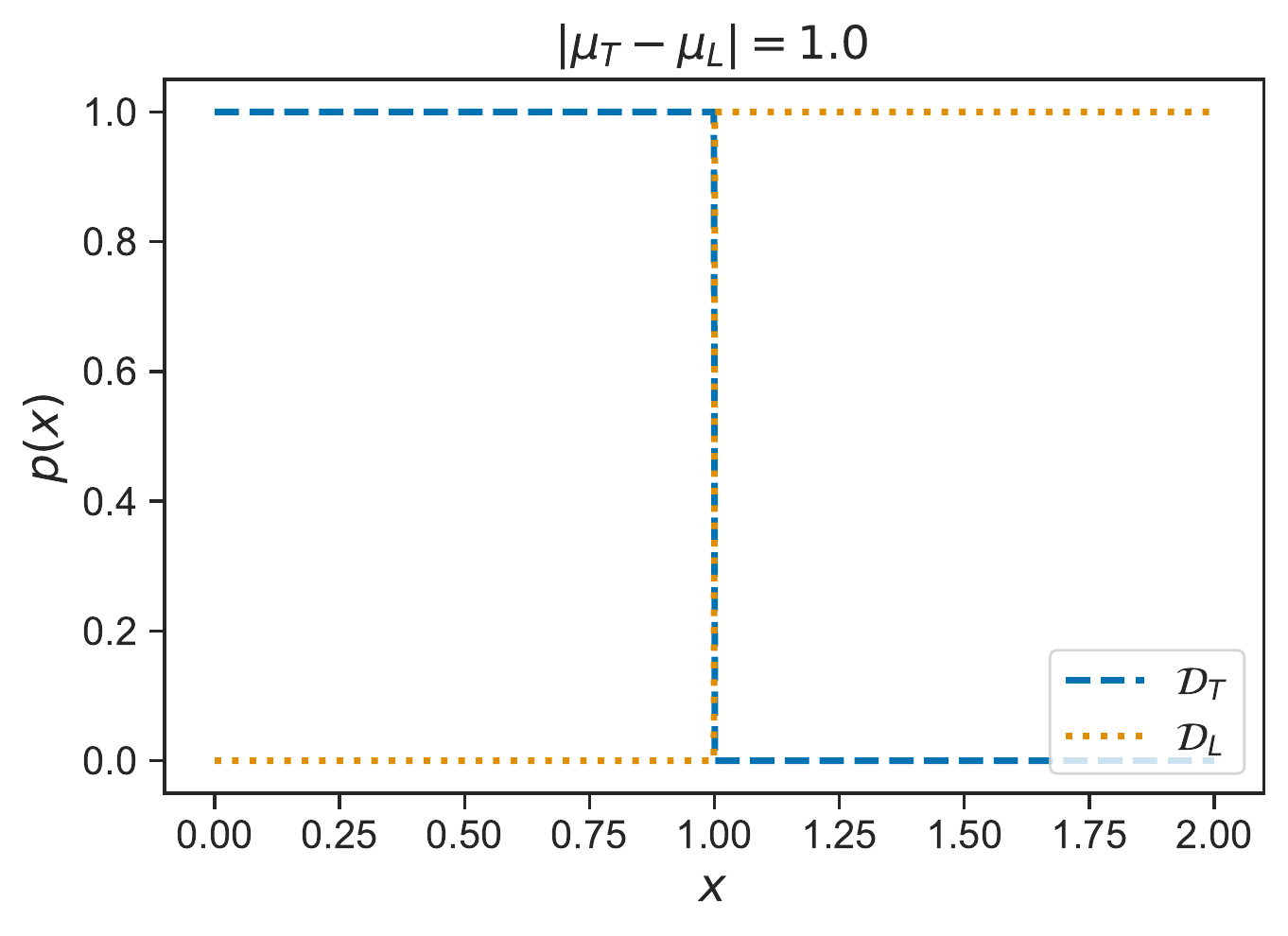}\\
  (a) & (c) & (e) \\
  \includegraphics[width=.33\linewidth]{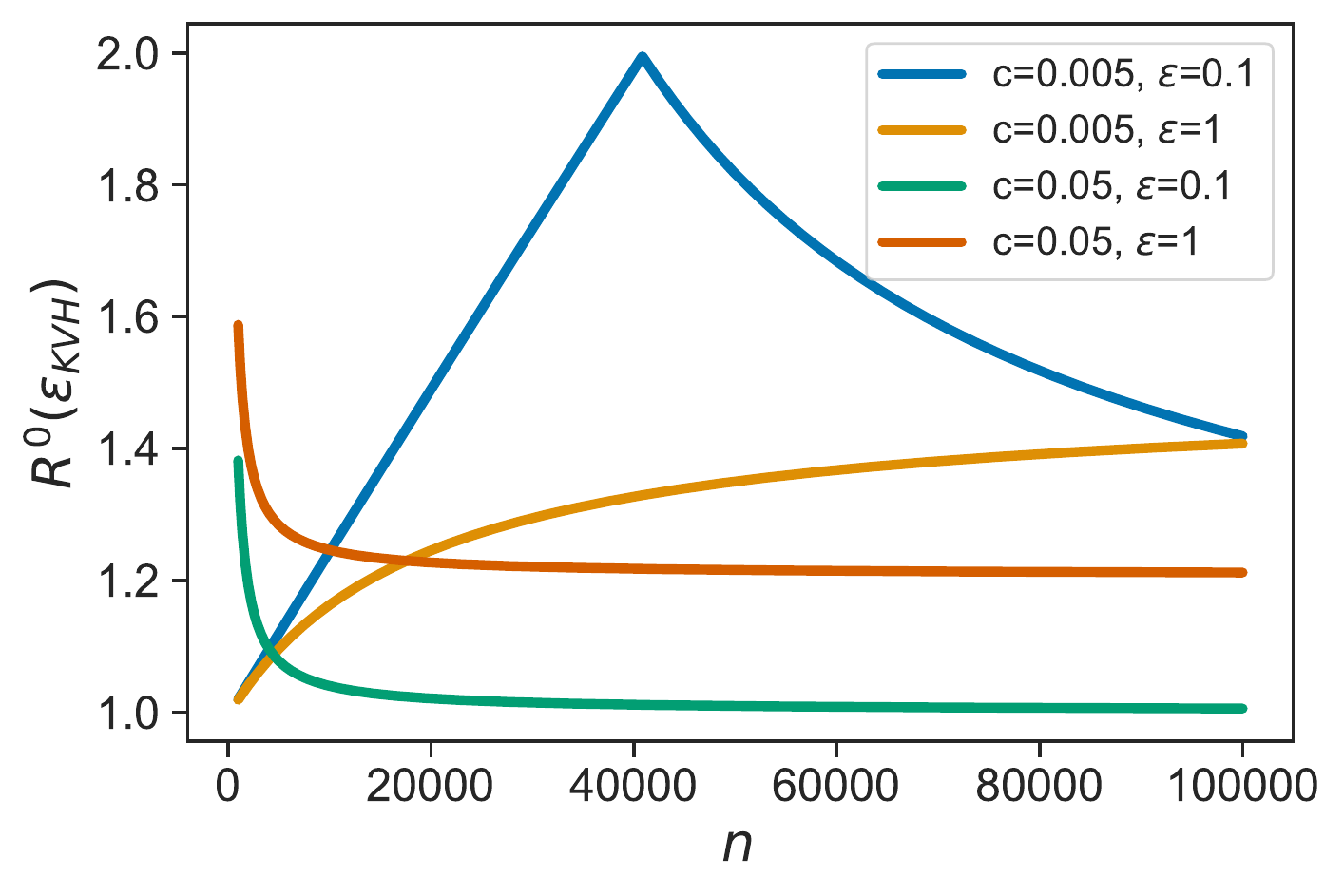} &  \includegraphics[width=.33\linewidth]{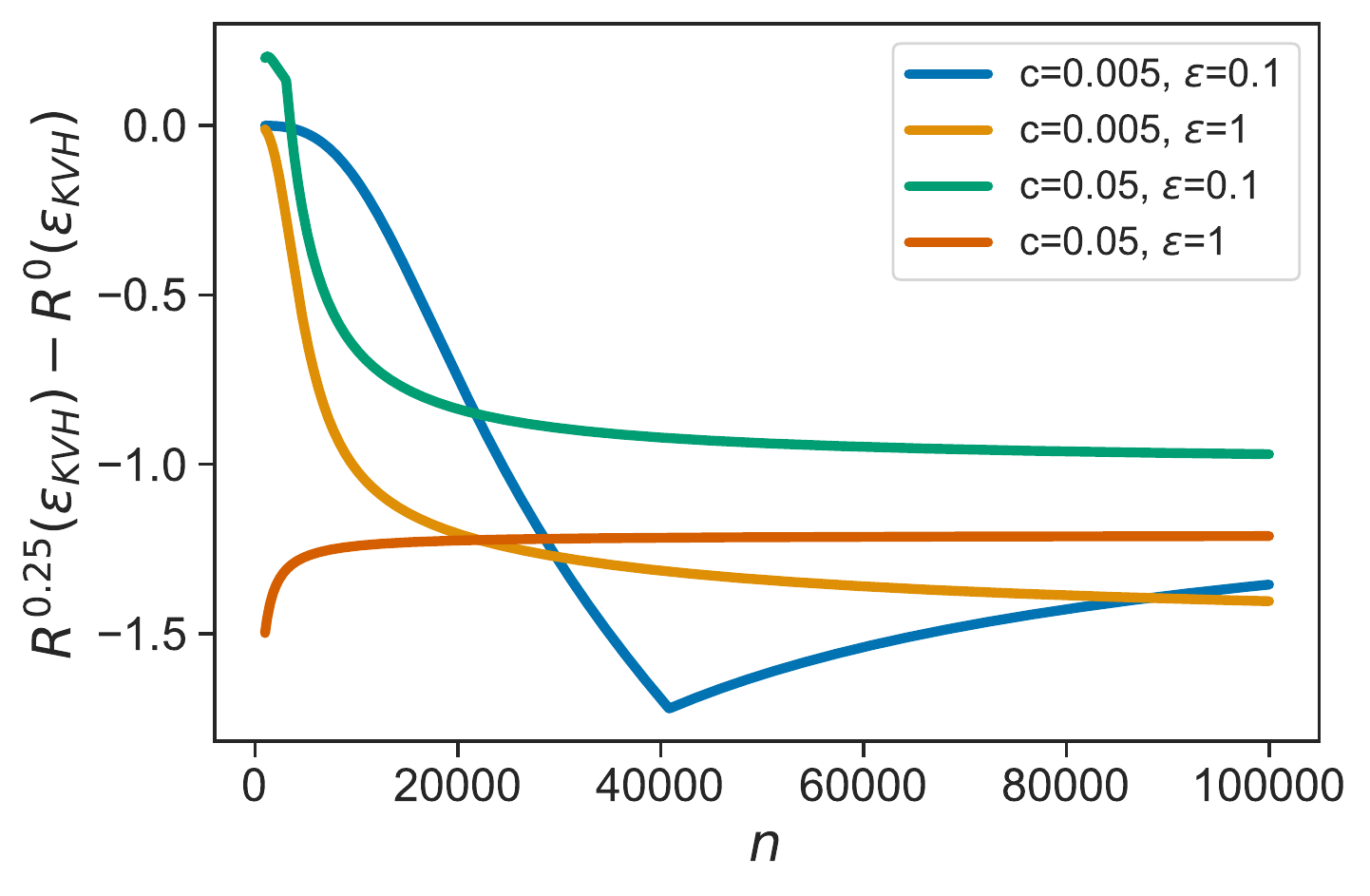} & \includegraphics[width=.33\linewidth]{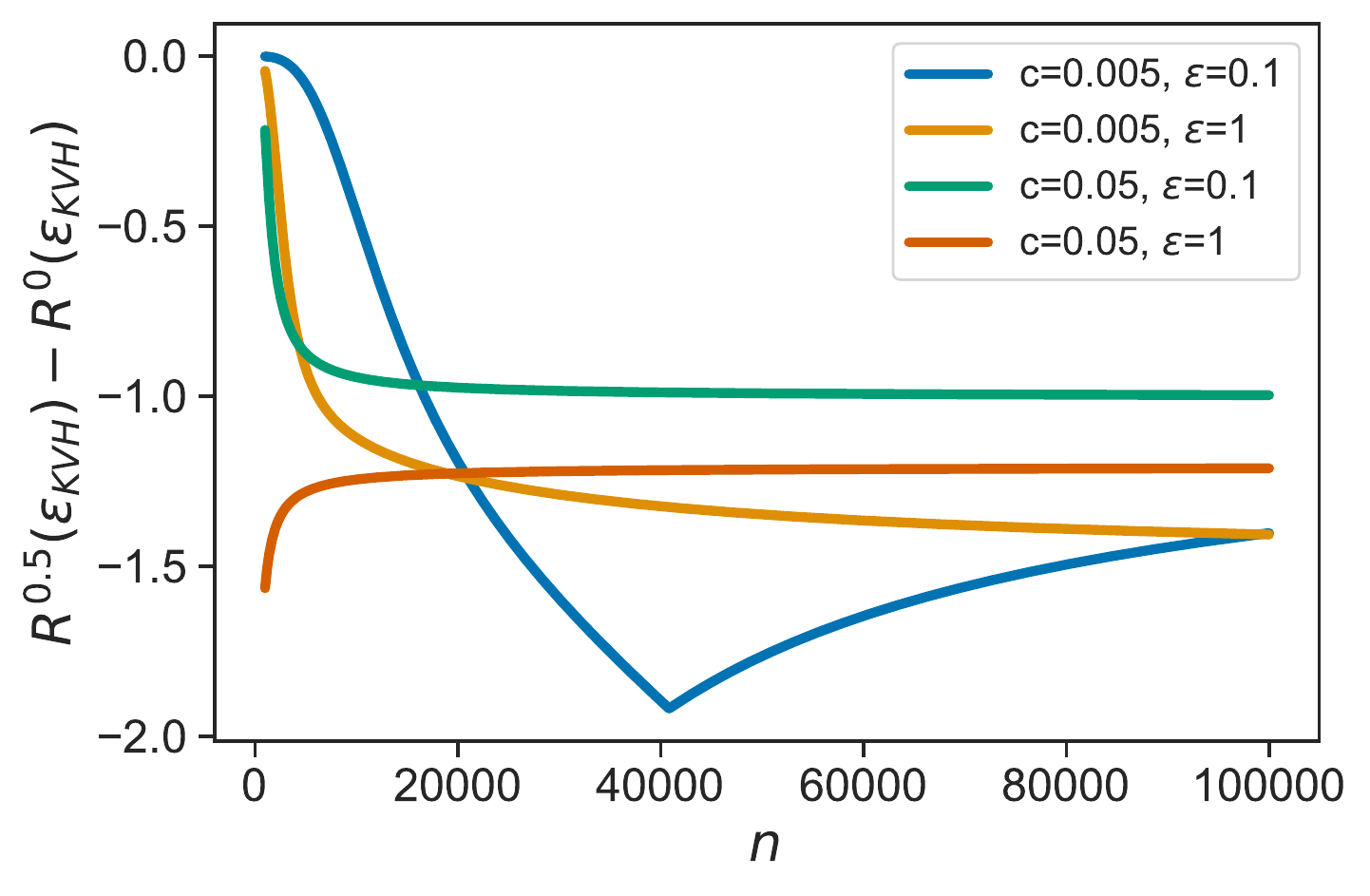}\\
  (b) & (d) & (f) \\
\end{tabular}
\caption{(Left column) Initial data distributions with no mean shift, and the KVH estimator's corresponding relative improvement. Small (middle column) and large (right column) mean shifts of the initial data distribution with $t=0.25$ and $t=0.5$ respectively, along with the KVH estimator's corresponding change in relative improvement.}
\label{fig:beta-R-hetero-means}
\end{figure*}

\section{Hybrid Estimator Applications} \label{sec:apps}
In this section, we demonstrate how more complex non-hybrid algorithms can be easily extended into the hybrid model using our hybrid esimator as a mean estimation primitive.
In particular, we implement a hybrid variant of the classic DP $K$-means algorithm~\cite{dwork2011firm} using the PWH hybrid estimator as a sub-component, then empirically evaluate its effectiveness.

The $K$-means problem is to partition $n$ $d$-dimensional real-valued observations $x_1,\dots, x_n$ into $K$ clusters $C_1,\dots, C_K$ such that the within-cluster sum of squares (WCSS) is minimized.
Denoting $\mu_k$ as the center of cluster $C_k$, this is formally: 
$$\argmin\limits_{C_1,\dots, C_K}\sum_{k=1}^K \sum_{x \in C_k} ||x-\mu_k||^2.$$
This problem is NP-hard, and thus heuristic algorithms are generally used.
The classic DP algorithm for this problem was designed for the TCM and analyzed in \cite{dwork2011firm}.
This algorithm partitions the total privacy budget across $\tau$ iterations, and each iteration refines the estimates of the clusters' centers.
Each iterative refinement assigns the observations to their nearest cluster, then updates each cluster's center to the mean of all points within it while carefully applying Laplace noise.

We extend this algorithm to the LM in a simple way.
First, LM users expend a portion of their privacy budget reporting their data to the curator with Laplace noise.
The curator uses their data analogously to the TCM case, except that in each iteration, LM users use a portion of their privacy budget to report the nearest cluster to them using randomized response -- this reduces bias in the cluster centers, relative to computing the nearest cluster directly based on their already-reported data.

Other DP $K$-means algorithms exist in both the TCM~\cite{nissim2007smooth, su2016differentially, barger2016k, nock2016k, balcan2017differentially, lu2020differentially} and LM~\cite{nissim2017clustering, sun2019distributed, xia2020distributed, stemmer2020locally} which improve on our two non-hybrid $K$-means algorithms.
However, the purpose of this section is to demonstrate how our hybrid estimator can be effectively leveraged in more complex applications.
Thus, we present our hybrid $K$-means algorithm in Appendix~\ref{alg:k-means}, which combines our simpler TCM and LM algorithms in the following straightforward way.
Each separate algorithm performs its iterative refinement as previously described.
Then, at the end of each iteration, the TCM and LM cluster center estimates are combined using the PWH estimator on each dimension.

We evaluate the hybrid algorithm in the following experiment, showing that it automatically achieves WCSS on-par with the best baseline.
The baselines here, analogous to our estimators' TCM-Only and Full-LM MSE baselines, are: the WCSS of the TCM variant using only TCM data, and the WCSS of the LM variant using all data.
The dataset used for evaluations is shown in Figure~\ref{fig:k-means}a: $4$ clusters of $2$-dimensional spherical Gaussian data with scale $\sigma\approx 0.028$ and $40,000$ points per cluster.
In Figure~\ref{fig:k-means}bc, across a range of total iterations $\tau$ and fractions of TCM users $0.1\%$ and $1\%$, we evaluate the mean WCSS values of each model's algorithm with $364$ trials.
The privacy budget for each algorithm is $\epsilon=7$; this relatively high budget is necessitated for the TCM and LM algorithms to achieve acceptable practical utility.
The regimes where each non-hybrid algorithm is better than the other is unclear a priori, and the results here show one example of each.
By simply combining the two using our hybrid estimator, the hybrid algorithm is able to maintain a WCSS approximately equal to the better of the two.

\begin{figure*}[h]
\centering
\begin{tabular}{ccc}
 \hspace*{-1cm}
  \includegraphics[width=.33\linewidth]{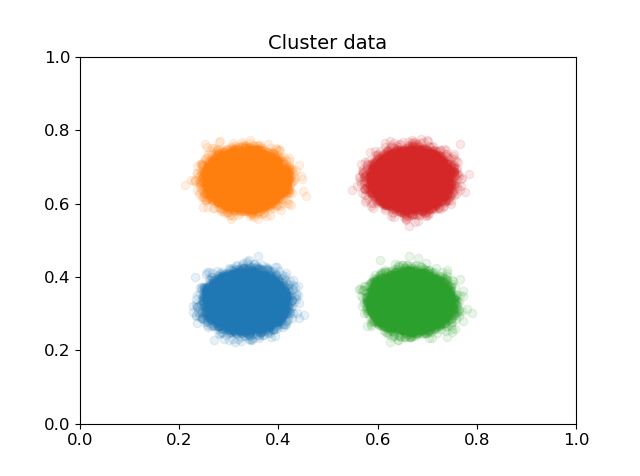} & \includegraphics[width=.33\linewidth]{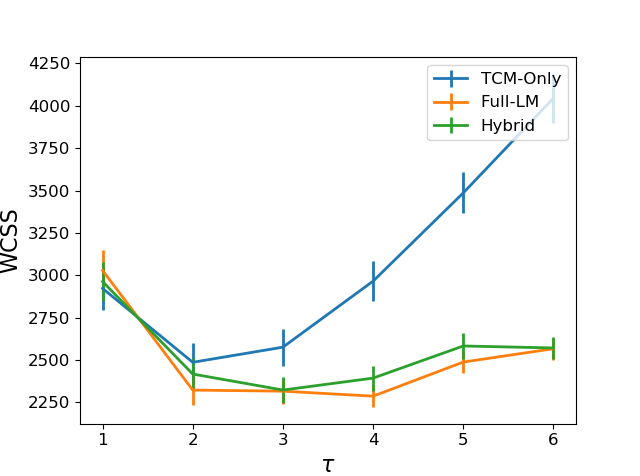} & \includegraphics[width=.33\linewidth]{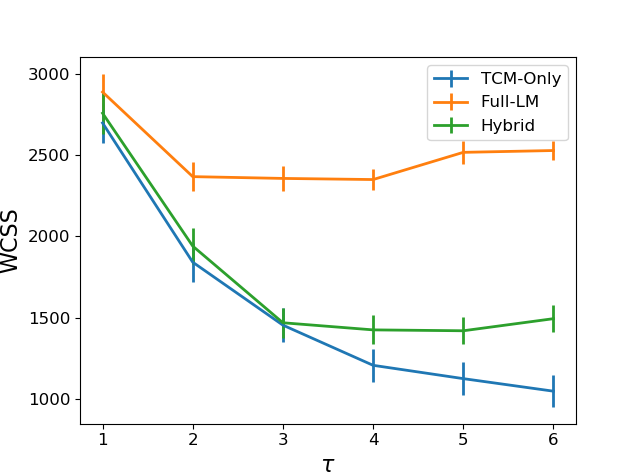} \\
  (a) & (b) & (c) \\
\end{tabular}
\caption{(a) Clustering dataset with $4$ clusters of $2$-d spherical Gaussians with $\sigma\approx 0.028$ and $40,000$ points per cluster. (b,c) WCSS values of each model's algorithm across a range of total iterations $\tau$, $0.1\%$ and $1\%$ fractions of TCM users respectively, and $\epsilon=7$.}
\label{fig:k-means}
\end{figure*}

\section{Privacy Amplification Via Inter-group Interaction} \label{sec:amp}
The benefit or necessity of inter-group interaction in the hybrid model is an active area of research.
Avent et al.~\cite{blender} showed experimentally that good utility is achievable by intelligently utilizing inter-group interaction.
In this work, we've shown mathematically that we can guarantee good utility for mean estimation with no inter-group interactivity.
However, both works only focus on inter-group interactivity's effect on a mechanism's \emph{utility} -- neither consider its effect on \emph{privacy}.
Each group is assumed to independently guarantee privacy, without considering how subsequent interaction and processing by the curator may effect the DP guarantee.
The post-processing property of DP ensures that such interaction and processing won't degrade privacy, but the question of whether it improves privacy was unstudied.
It is precisely this effect on privacy that we introduce and examine in this section.
We find that for our hybrid mean estimators, the privacy guarantee against certain adversaries can be significantly improved.

We are specifically interested in users' privacy against adversaries who can view the output of the curator's computation (i.e., output-viewing adversaries). 
This is the classic adversary model that the TCM protects against.
The LM protects against a larger class of adversaries: the output-viewing adversaries, as well as against the curator itself.
However, the LM's singular DP guarantee doesn't distinguish between these adversary types.
In the hybrid model, each groups' DP guarantee may be overly-conservative against output-viewing adversaries since it doesn't account for the curator's joint processing of the LM users' reports -- which each include their own privacy noise -- in conjunction with the TCM group's privacy noise.
Thus, we investigate users' DP guarantee against output-viewing adversaries as a result of: 1) the combined privacy noise from both groups, in conjunction with 2) the inter-group interaction strategy of the curator.
We show that these two components together can serve to amplify users' privacy against this adversary class.
This provides a two-tier DP guarantee for LM users -- their standard DP guarantee against the curator, and an improved guarantee against output-viewing adversaries -- and an improved DP guarantee for TCM users.
To make this concrete, we first analyze our hybrid estimator family and show how its non-interactive strategy can amplify privacy.
We then describe why BLENDER's~\cite{blender} interaction strategy does not provide such amplification.
Together, these examples highlight the value of looking at the effects of inter-group interaction not only on utility, but also on privacy.

\paragraph*{Hybrid Mean Estimator Amplification}
Recall that the hybrid estimator family from Definition~\ref{def:hybrid} utilizes no inter-group interaction -- i.e., the curator only outputs once: after it has received all the LM users reports, computed both groups' estimates, then combined them.
For adversaries that can only view the output of this curator, the combined noise from \emph{all} the LM users and the TCM group can serve to improve the DP guarantee.
To see this, we re-write the estimator as
\begin{align*}
\tilde{\mu}_{H}(w) &= \underbrace{\left(\frac{w}{cn} \sum_{i \in T} x_i + \frac{1-w}{(1-c)n} \sum_{i \in L} x_i\right)}_{\text{non-private hybrid mean estimator}} \\
&+ \underbrace{\left(wY_T + \frac{1-w}{(1-c)n} \sum_{i \in L} Y_{L,i}\right)}_{\text{joint privacy noise}}.
\end{align*}
Thus, this joint privacy noise is providing some DP guarantee for the mechanism as a whole, rather than individual noises protecting the individual groups.

There is one caveat: the TCM users' noise is provided by the curator and never revealed to them, but the LM users each provide their own noise.
DP requires that the privacy noise not be known to an adversary; any noise that is known cannot be considered towards the DP guarantee.
Here, we assume LM users are semi-honest -- i.e., they apply the specified mechanism properly to their data, but they know the privacy noise they add.
Thus, LM user $i$'s knowledge of their own privacy noise weakens the the joint noise term by an additive $\frac{1-w}{(1-c)n} Y_{L,i}$ amount.
Furthermore, they may choose to form coalitions with other users and share this knowledge to adversarially weaken the joint privacy noise term.
The largest such coalition, denoted as $A$, reduces the joint privacy noise by $\frac{1-w}{(1-c)n} \sum_{i \in A} Y_{L,i}$.
Excluding the largest such coalition's noise enables the remaining joint privacy noise to be analyzed for a DP guarantee.

The DP guarantee from the remaining joint noise depends on the privacy mechanisms used by the TCM group and each LM user.
For instance, the $\epsilon$-DP Laplace mechanism would yield a joint noise term which guarantees $\epsilon^\prime$-DP where $\epsilon^\prime = \epsilon$ -- i.e., it wouldn't enable any privacy amplification (see proof in Appendix~\ref{app:lap-amp-proof}).
Alternatively, consider the Gaussian mechanism, where the curator adds $Y_T \sim \mathcal{N}(0,s_T^2)$ and each LM user $i$ adds $Y_{L,i} \sim \mathcal{N}(0,s_L^2)$, where $s_T^2$ and $s_L^2$ are calibrated to ensure $(\epsilon,\delta)$-DP for both groups.
Now, analyzing this joint noise provides the following amplified DP guarantee against output-viewing adversaries.
\begin{thm} \label{thm:amp-gaussian}
Assume the curator adds Gaussian noise of variance $s_T^2$ to provide an $(\epsilon, \delta)$-DP guarantee for the TCM group, and that each LM user adds Gaussian noise of variance $s_L^2$ to provide an $(\epsilon, \delta)$-DP guarantee for themselves.
Furthermore, assume that the largest adversarial coalition is of size $|A|$.
Define 
$$s^{\prime 2} =w^2s_T^2 + \left(\frac{1-w}{(1-c)n}\right)^2 |L\setminus{A}| s_L^2.$$
The users' $\epsilon^\prime$-DP guarantee against output-viewing adversaries is given by:
\begin{align*}
\epsilon^\prime = \frac{\sqrt{2\ln(1.25/\delta)}m}{ns^\prime} \cdot \begin{cases}
          \frac{w}{c},& w \le c \\[0em]
          \frac{1-w}{1-c},& \text{otherwise}
         \end{cases}
\end{align*}
\end{thm}
\begin{proof}\renewcommand{\qedsymbol}{}
See Appendix~\ref{app:gm-amp-proof}.
\end{proof}

For practical applications, even with a moderate fraction of adversarial LM users, this amplification can be significant.
To make this concrete, consider the UC salary dataset used in the previous experiments.
Suppose we compute the KVH estimator with each group using the Gaussian mechanism with $\epsilon=1$ and $\delta=10^{-7}$.
In Figure~\ref{fig:gm-amp}, we plot the users' amplified $\epsilon^\prime$ value across $c \in [0.1\%, 10\%]$ as well as across the fraction of LM users assumed to be adversarial.
\begin{figure}[h]
\centering
  \includegraphics[width=.9\linewidth]{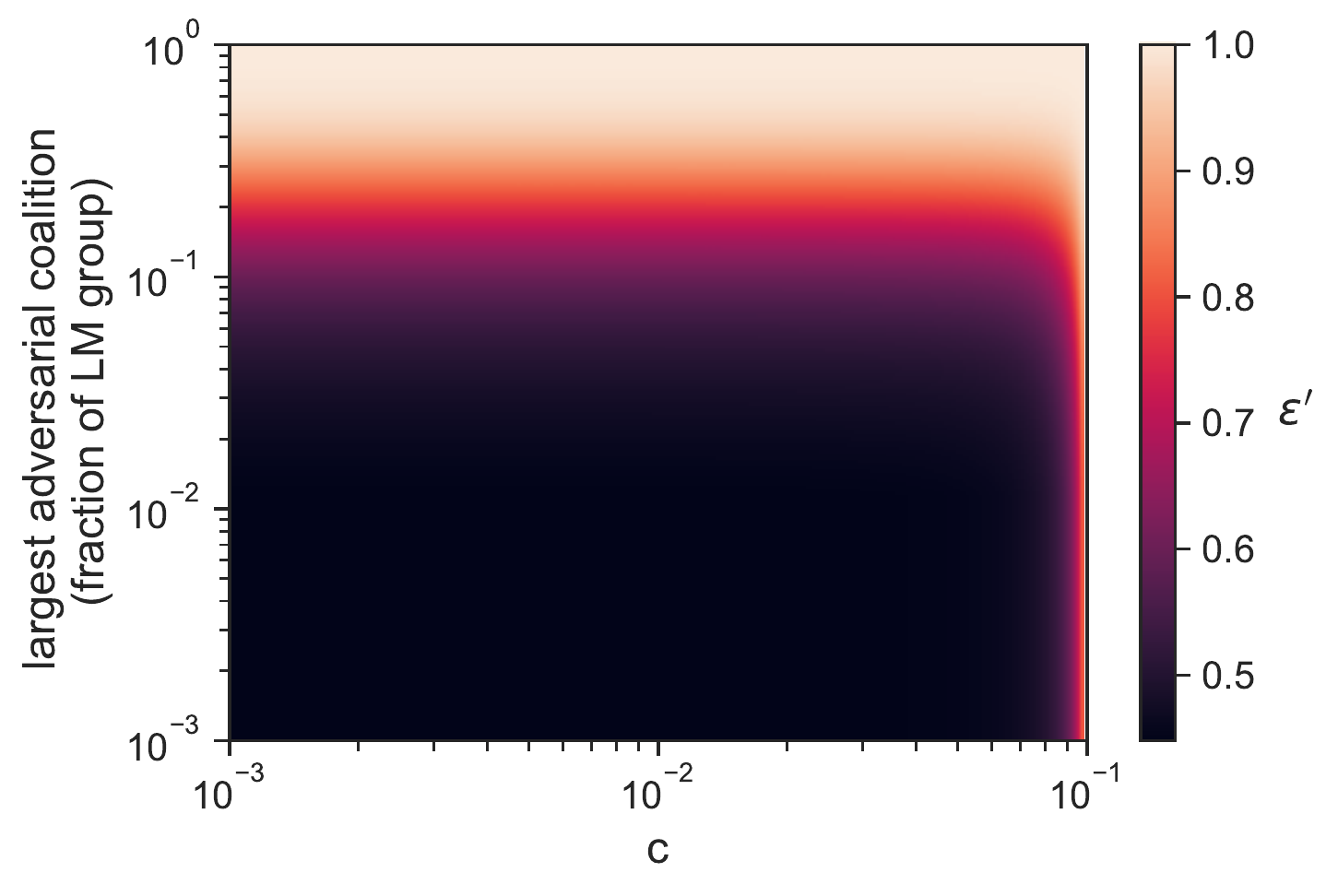} \\
\caption{The amplified $(\epsilon^\prime, \delta)$-DP guarantee when the $(1, 10^{-7})$-DP Gaussian mechanism is used in the KVH estimator.}
\label{fig:gm-amp}
\end{figure}

\paragraph*{BLENDER Amplification}
Now that we've shown how a \emph{lack} of inter-group interaction can facilitate privacy amplification, we turn our focus to BLENDER's inter-group interactions strategy for improving utility.
As described in Section~\ref{sec:related-work}, BLENDER takes advantage of the TCM group by having it identify the heavy hitters.
The TCM group then passes the identified heavy hitters on to the LM users, who perform frequency estimation.
The curator then combines the LM users' reports and outputs the heavy hitters along with their frequencies.

One might be tempted to analyze this final output for an amplified DP guarantee.
However, the initial output of the curator -- the privatized list of heavy hitters from the TCM group -- has already been released to all LM users.
Unless \emph{all} LM users are non-adversarial, the TCM users gain no further benefit from the incorporation of the LM users' privacy noise.
Conversely, the LM users may experience privacy amplification through the combination of their locally-added noise; however, this is solely due to intra-group interaction in the LM, and related topics are currently an active area of research on the LM~\cite{bittau2017prochlo, cheu2019distributed, balle2019differentially, erlingsson2019amplification, ghazi2019scalable, balcer2019separating}.

\section{Conclusions} \label{sec:conclusion}
Under the hybrid model of differential privacy, which permits a mixture of the trusted-curator and local trust models, we explored the problem of accurately estimating the mean of arbitrary distributions with bounded support.
We designed a hybrid estimator for the joint mean of both trust models' users and derived analytically exact finite sample utility expressions -- applicable even when the trust models may have different user data distributions.

When the trust models have the same distribution and the curator knows its variance, we proved that our estimator is able to always achieve higher utility than both baseline estimators in the two individual trust models.
When the variance is unknown, for many practical parameters, we showed that our hybrid estimator achieves utility nearly as good as when the variance is known.
For both cases, we evaluated our estimator on realistic datasets and found that it achieves good utility in practice.

By designing a hybrid variant of the classic differentially private $K$-means algorithm, we showed how more complex hybrid algorithms can be designed by using the hybrid estimator as a sub-component.
Experimentally, we found that this hybrid algorithm automatically achieves utility on-par with the best of its corresponding two non-hybrid algorithms, even though it is unclear a priori when each non-hybrid algorithm is better.

Finally, we introduced a notion of privacy amplification that arises in the hybrid model due to interaction between the two underlying trust models.
We derived the privacy amplification that our hybrid estimator provides all users, and discussed when other hybrid mechanisms may fail to achieve amplification.

\paragraph*{Acknowledgements}
This work was supported by NSF grants No.~1755992, 191615, and~1943584, a VMWare fellowship,
and a gift from Google.

The work was partially done while the authors B.~A. and A.~K. were at the ``Data Privacy: Foundations and Applications'' Spring 2019 program held at the Simons Institute for the Theory of Computing, UC Berkeley. 

\bibliographystyle{plain}
\bibliography{bibliography}

\appendix
\section{Estimator MSE Proofs}

\paragraph*{Proof of Lemma~\ref{lem:tcm-only-mse}} \label{proof:tcm-only-mse}
\begin{proof}[\unskip\nopunct]
\begin{equation*}
\small
\begin{split}
\V[\tilde{\mu}_T - \hat{\mu}] &= \V\left[\frac{1}{cn} \sum_{i \in T} x_i + Y_T - \frac{1}{n}\sum_{i \in [n]} x_i \right] \\
&= \V\left[\frac{1}{cn} \sum_{i \in T} x_i - \frac{1}{n}\sum_{i \in T} x_i - \frac{1}{n}\sum_{i \in L} x_i + Y_T \right] \\
&= \frac{(1-c)^2}{cn}\sigma_T^2 + \frac{1-c}{n}\sigma_L^2 + s_T^2.
\end{split}
\end{equation*}
Thus,
\begin{equation*}
\small
\begin{split}
\mathcal{E}_T &= \E[(\tilde{\mu}_T - \hat{\mu})^2] \\
&= \V[\tilde{\mu}_T - \hat{\mu}] + \E[\tilde{\mu}_T - \hat{\mu}]^2 \\
&= \frac{(1-c)^2}{cn}\sigma_T^2 + \frac{1-c}{n}\sigma_L^2 + s_T^2 + (\mu_T - \mu)^2.\qedhere
\end{split}
\end{equation*}
\end{proof}

\paragraph*{Proof of Lemma~\ref{lem:full-lm-mse}} \label{proof:full-lm-mse}
\begin{proof}[\unskip\nopunct]
\begin{equation*}
\small
\begin{split}
\mathcal{E}_F &= \E[(\tilde{\mu}_F - \hat{\mu})^2] \\
&= \V[\tilde{\mu}_F - \hat{\mu}] + \underbrace{\E[\tilde{\mu}_F - \hat{\mu}]^2}_0 \\
&= \V\left[\frac{1}{n} \sum_{i \in [n]} (x_i + Y_{L,i}) - \frac{1}{n}\sum_{i \in [n]} x_i\right] \\
& = \frac{s_L^2}{n}.\qedhere
\end{split}
\end{equation*}
\end{proof}

\paragraph*{Proof of Lemma~\ref{lem:lm-only-mse}} \label{proof:lm-only-mse}
\begin{proof}[\unskip\nopunct]
\begin{equation*}
\small
\begin{split}
\V&[\tilde{\mu}_L - \hat{\mu}] \\
&= \V\left[\frac{1}{(1-c)n} \sum_{i \in L} (x_i + Y_{L,i}) - \frac{1}{n}\sum_{i \in [n]} x_i \right] \\
&= \V\left[\frac{c}{(1-c)n} \sum_{i \in L} x_i - \frac{1}{n}\sum_{i \in T} x_i + \frac{1}{(1-c)n} \sum_{i \in L} Y_{L,i}\right] \\
&= \frac{c^2}{(1-c)n}\sigma_L^2 + \frac{c}{n}\sigma_T^2 + \frac{1}{(1-c)n}s_L^2.
\end{split}
\end{equation*}
Thus,
\begin{equation*}
\small
\begin{split}
\mathcal{E}_L &= \E[(\tilde{\mu}_L - \hat{\mu})^2] \\
&= \V[\tilde{\mu}_L - \hat{\mu}] + \E[\tilde{\mu}_L - \hat{\mu}]^2 \\
&= \frac{c^2}{(1-c)n}\sigma_L^2 + \frac{c}{n}\sigma_T^2 + \frac{1}{(1-c)n}s_L^2 + (\mu_L - \mu)^2.\qedhere
\end{split}
\end{equation*}
\end{proof}

\paragraph*{Proof of Lemma~\ref{lem:hybrid-mse}} \label{proof:hybrid-mse}
\begin{proof}[\unskip\nopunct]
\begin{equation*}
\small
\begin{split}
\V&[\tilde{\mu}_H(w) - \hat{\mu}] \\
&= \V[w\tilde{\mu}_T + (1-w)\tilde{\mu}_L - \hat{\mu}] \\
&= \V[w\tilde{\mu}_T - c\hat{\mu}_T + (1-w)\tilde{\mu}_L - (1-c)\hat{\mu}_L] \\
&= \frac{(w-c)^2}{cn}\sigma_T^2 + \frac{(w-c)^2}{(1-c)n}\sigma_L^2 + w^2 s_T^2 + \frac{(1-w)^2}{(1-c)n}s_L^2.
\end{split}
\end{equation*}
Thus,
\begin{equation*}
\small
\begin{split}
\mathcal{E}_{H}(w) &= \E[(\tilde{\mu}_{H}(w) - \hat{\mu})^2] \\
&= \V[\tilde{\mu}_{H}(w) - \hat{\mu}] + \E[\tilde{\mu}_{H}(w) - \hat{\mu}]^2 \\
&= \frac{(w-c)^2}{cn}\sigma_T^2 + \frac{(w-c)^2}{(1-c)n}\sigma_L^2 + w^2 s_T^2 + \frac{(1-w)^2}{(1-c)n}s_L^2 \\
&\quad + (w\mu_T + (1-w)\mu_L - \mu)^2.\qedhere
\end{split}
\end{equation*}
\end{proof}

\section{Proof of Corollary~\ref{cor:relaistic-improvement}} \label{proof:realistic-improvement}
Note the following for upper-bounding $R(\mathcal{E}_{KVH})$.
Popoviciu's inequality \cite{popoviciu1935equations} states that a random variable bounded in $[a,b]$ has variance at most $(b-a)^2/4$.
For our purposes, this ensures $\sigma^2 \le m^2/4$.
For real-world use cases, it is realistic to constrain $\epsilon$ to the ``high-privacy'' regime of $\epsilon \le 1$.
Thus, with $\epsilon \le 1$ and $\sigma^2 \le m^2/4$, we have $0 \le\epsilon^2 \sigma^2 / m^2 \le 1/4$.
Let $y = \epsilon^2 \sigma^2 / m^2$.
Now, we upper-bound the improvement ratio as follows.
$$\small R(\mathcal{E}_{KVH}) \le \frac{2(2-c)m^2}{2m^2 - (1-c)\epsilon^2\sigma^2} = \frac{2(2-c)}{2 - (1-c)y} \le 16/7,$$
where the final inequality stems from constrained maximization across $c \in [0,1]$ and $y \in [0,1/4]$ (justified in the above note).

A lower-bound is given by the following concrete instance.
Let $m=1$, $\epsilon=1$, $\sigma^2=1/4$, and $c=\frac{1}{18}\left(1+\sqrt{\frac{288+n}{n}}\right)$.
Then, as $n \rightarrow \infty$, we have that $R(\mathcal{E}_{KVH})$ converges to $17/8$.

\section{PWH Utility (continued)} \label{app:pwh-uc-salary}

Figure~\ref{fig:uc-salary-pwh} presents heatmaps of $R(\mathcal{E}_{PWH})$ and $r(\mathcal{E}_{PWH})$ for the UC salaries dataset across the same parameters as before ($c \in [0.1\%, 10\%]$ and $\epsilon \in [0.1, 10]$).
We find that $R(\mathcal{E}_{PWH})$ achieves a value of slightly greater than $1$ across a large portion of the space.
The results here tell a similar story to that of Figure~\ref{fig:beta-pwh}.
Most of the space has $R(\mathcal{E}_{PWH})$ values above $1$, and even approaching $2$ in narrow region.
There is also a small region at the large $c$ values where the relative improvement drops below $0.5$.
The majority of the space has $r(\mathcal{E}_{PWH})$ between $10$ and $100$, although it includes a region at the high $\epsilon$ values where this relative improvement exceeds $1,500$.

\begin{figure}[h]
\centering
\begin{tabular}{c}
  \includegraphics[width=.8\linewidth]{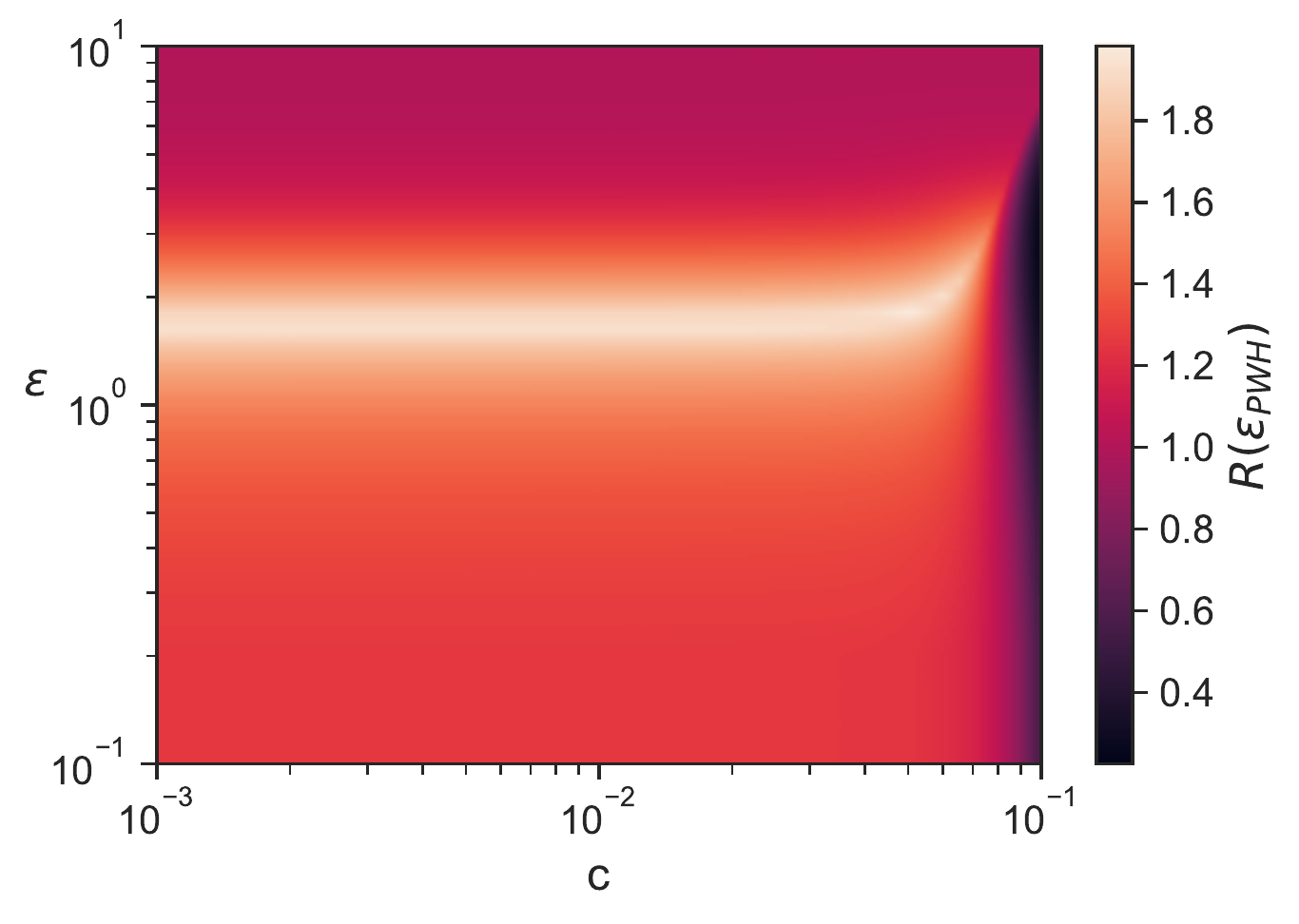} \\
  (a) \\
   \includegraphics[width=.8\linewidth]{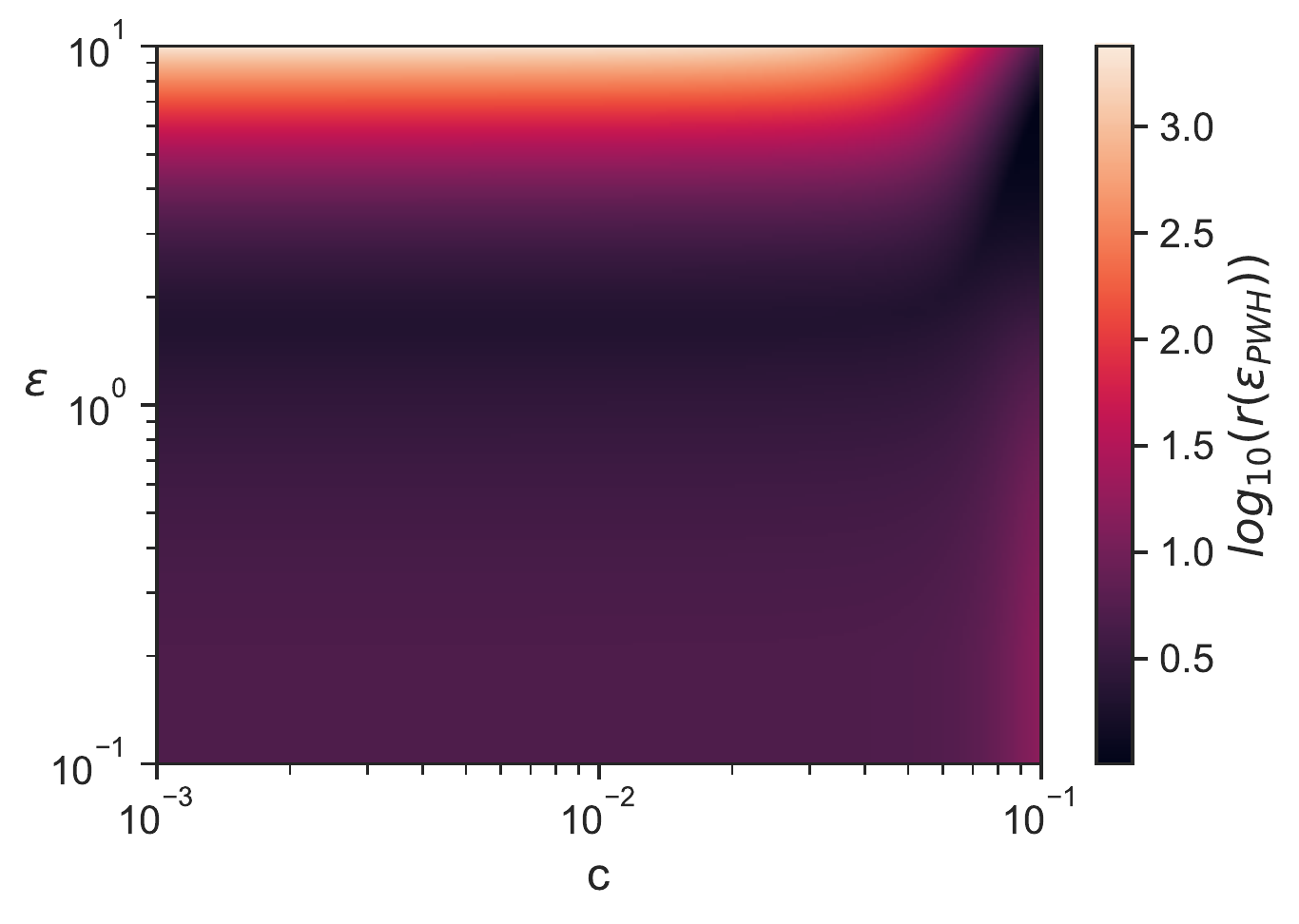} \\
(b) \\
\end{tabular}
\caption{The relative improvements $R(\mathcal{E}_{PWH})$ (a) and  $r(\mathcal{E}_{PWH})$ (b) across a range of $c$ and $\epsilon$ values. Note the log scale on (b).}
\label{fig:uc-salary-pwh}
\end{figure}

\section{Privacy Amplification Proofs}

\paragraph*{$\epsilon$-DP Laplace Mechanism Amplification} \label{app:lap-amp-proof}

Here, we show that the unweighted sum of user's reports, each privatized by their own $\epsilon$-DP Laplace mechanism, only provides a $\epsilon^\prime$-DP joint guarantee of $\epsilon^\prime = \epsilon$.
We note that a convex weighting of the terms in this sum yields the same joint guarantee, proving our claim.
To formalize this, we define
$$\tilde{S}_n := \sum_{i\in[n]} (x_i + Y_i) = \sum_{i\in[n]} x_i + \sum_{i\in[n]}Y_i$$
where $x_i \in [0,m]$, $Y_i \sim \text{Lap}(b)$, and $b = m/\epsilon$ for each user $i$.
Denote $S_n := \sum_{i\in[n]} x_i$ and $Y := \sum_{i\in[n]}Y_i$.
We show here that the joint noise $Y$ provides $\epsilon^\prime$-DP for $S_n$ against output-viewing adversaries, where $\epsilon^\prime = \epsilon$.

First, note that each user at least has the $\epsilon$-DP guarantee via their own privacy noise.
Thus, by the post-processing property of DP, we have $\epsilon^\prime \le \epsilon$ as a trivial upper-bound.
Note that if $\epsilon \le \epsilon^\prime$ without any adversarial users, then our upper-bound implies $\epsilon \le \epsilon^\prime$ with an arbitrary number of adversarial users.
Therefore, we assume w.l.o.g. that no users are adversarial.

Towards lower-bounding $\epsilon^\prime$, note that the characteristic function of each $Y_i$ is
$$\varphi_{Y_i}(t) = \frac{1}{1+b^2t^2}.$$
Then the characteristic function of $Y$ is
$$\varphi_{Y}(t) = \prod_{i\in[n]} \varphi_{Y_i}(t) = \left(\frac{1}{1+b^2t^2}\right)^n.$$
$Y$'s probability density function, $p_Y(x)$, can be recovered from the characteristic function~\cite{billingsley2008probability} via the inverse Fourier transform as
\begin{align*}
p_Y(x) &= \frac{1}{2\pi}\int_\mathbb{R} e^{itx} \overline{\varphi_{Y}(t)}dt \\
&= \frac{1}{2\pi}\int_\mathbb{R} \frac{e^{itx}}{(1+b^2t^2)^n}dt \\
&=  \frac{2^{\frac{1}{2} - n}}{\sqrt{\pi} b^{\frac{1}{2} + n} \Gamma(n)} \mathcal{K}_{\frac{1}{2} - n}\left(\frac{|x|}{b}\right)|x|^{n - \frac{1}{2}},
\end{align*}
where $\overline{\varphi_{Y}(t)}$ is the complex conjugate of $\varphi_{Y}(t)$ and $\mathcal{K}_{\cdot}(\cdot)$ is the modified Bessel function of the second kind~\cite{abramowitz1968handbook}.\\

\noindent For $\epsilon^\prime$-DP, noting that $\Delta_1 S_n = m$, we must bound
$$-\epsilon^\prime \le \max_{k \in [-m,m]}\log\left(\frac{p_Y(x)}{p_Y(x+k)}\right) \le \epsilon^\prime.$$
Consider the instance where $k=m$ and $x\rightarrow \infty$:
\begin{align*}
&\lim_{x\rightarrow\infty} \log\left(\frac{p_Y(x)}{p_Y(x+k)}\right) \\
&= \lim_{x\rightarrow\infty} \log\left(\frac{\mathcal{K}_{\frac{1}{2} - n}\left(\frac{x}{b}\right)}{\mathcal{K}_{\frac{1}{2} - n}\left(\frac{x+m}{b}\right)} \left(\frac{x}{x+m}\right)^{n - \frac{1}{2}}\right) \\
&= \frac{m}{b}.
\end{align*}

\noindent By the definition of $b$, we have $\frac{m}{b} = \epsilon$.
Therefore
\begin{align*}
\epsilon &= \lim_{x\rightarrow \infty}\log\left(\frac{p(x)}{p(x+k)}\right) \\
&\le \max_{k \in [-m,m]}\log\left(\frac{p(x)}{p(x+k)}\right) \\
&\le \epsilon^\prime.
\end{align*}
Thus, we conclude that $\epsilon^\prime = \epsilon$. \qed

\paragraph*{$(\epsilon, \delta)$-DP Gaussian Mechanism Amplification} \label{app:gm-amp-proof}
Denote the non-private hybrid mean estimator as
$$\hat{\mu}_H(w) := \frac{w}{cn} \sum_{i \in T} x_i + \frac{1-w}{(1-c)n} \sum_{i \in L} x_i$$
and the joint privacy noise (without the largest adversarial coalition) as
$$Y := wY_T + \frac{1-w}{(1-c)n} \sum_{i \in L \setminus{A}} Y_{L,i}.$$
We first compute the sensitivity
$$\Delta_2 \hat{\mu}_H(w) = \max \left\Vert \hat{\mu}_H(w) - \hat{\mu}_H^\prime(w) \right\Vert_2,$$
where $\hat{\mu}_H(w)$ is the estimator on any dataset $D = T \cup L$ and $\hat{\mu}_H^\prime(w)$ is the estimator on any neighboring dataset $D^\prime = T^\prime \cup L^\prime$ differing in the data of at most one user.
If the data of one $T$ user is changed, then $\max \left\Vert \hat{\mu}_H(w) - \hat{\mu}_H^\prime(w) \right\Vert_2 \le \frac{wm}{cn}$.
If instead the data of one $L$ user is changed, then $\max \left\Vert \hat{\mu}_H(w) - \hat{\mu}_H^\prime(w) \right\Vert_2 \le \frac{(1-w)m}{(1-c)n}$.
Note that $\frac{wm}{cn} \le \frac{(1-w)m}{(1-c)n}$ when $w \le c$.
Thus, we have 
$$\Delta_2 \hat{\mu}_H(w) =
\begin{cases}
\frac{wm}{cn},& w \le c \\
\frac{(1-w)m}{(1-c)n},& \text{otherwise}
\end{cases}.$$

Next, let $Y_T \sim \mathcal{N}(0,s_T^2)$ and $Y_{L,i} \sim \mathcal{N}(0,s_{L,i}^2)$ such that $s_T$ satisfies $(\epsilon, \delta)$-DP for the TCM group and $s_{L,i}$ satisfies $(\epsilon, \delta)$-DP for each LM user $i$.
By the well-known properties of Gaussians, the weighted combination of Gaussians is also a Gaussian, as 
$$Y \sim \mathcal{N}\left(0,\ s^{\prime 2} := w^2s_t^2 + \left(\frac{1-w}{(1-c)n}\right)^2|L\setminus A|s_L^2\right).$$

Recall that the classic Gaussian mechanism~\cite{dwork2006our} guarantees $(\epsilon, \delta)$-DP for a function $f$ with sensitivity $\Delta_2 f$ by adding noise from $\mathcal{N}(0,s^2)$ such that $s^2 = 2\log(1.25/\delta)\Delta_2 f^2 / \epsilon^2$.
Applying this result to our problem with a fixed $\delta^\prime = \delta$ and solving $\epsilon^\prime$, we have
\begin{align*}
\small
\epsilon^\prime &= \sqrt{\frac{2\ln(1.25/\delta)\Delta_2\hat{\mu}_H^2(w)}{s^{\prime 2}}} \\
&= \frac{\sqrt{2\ln(1.25/\delta)}m}{ns^\prime} \cdot \begin{cases}
          \frac{w}{c},& w \le c \\[0em]
          \frac{1-w}{1-c},& \text{otherwise}
         \end{cases}. \qed
\end{align*}

\newpage
\section{Hybrid K-means Pseudocode} \label{alg:k-means}

\begin{algorithm}
\DontPrintSemicolon
\caption{Hybrid-DP $K$-means} 
\KwIn{TCM users $T$, LM users $L$, data range $m$, data dimension $d$, num. clusters $K$, num. iterations $\tau$}
Initialize centers of clusters $C_1, \dots, C_K$\;
$b_T \leftarrow \frac{(md+1)\tau}{\epsilon}$ and $b_L \leftarrow \frac{md(\tau+1)}{\epsilon}$\;
Each $i \in T$ reports $\tilde{x}_i \leftarrow x_i$ to the curator\;
Each $i \in L$ reports $\tilde{x}_i \leftarrow x_i + Y_{L,i}$ to the curator, $Y_{L,i} \sim \text{Lap}^d(b_L)$\;
\For{$t \leftarrow 1\dots\tau$}{
Assign each $\tilde{x}_i$ from $T$ to closest cluster non-privately\;
Assign each $\tilde{x}_i$ from $L$ to closest cluster with prob. $\frac{\exp(\epsilon/(\tau+1))-1}{K+\exp(\epsilon/(\tau+1))-1}$; to a uniformly random cluster otherwise\;
\For{$k \leftarrow 1 \dots K$}{
Count $T$ users in cluster $k$ with DP : $\tilde{N}_{T} \leftarrow |C_k \cap T| + Y_1,\ Y_1 \sim \text{Lap}(b_T)$\;
Compute mean of all $T$ users' data in cluster $k$ with DP: $\tilde{\mu}_{T} \leftarrow \frac{1}{\tilde{N}_{T}}\left(\sum_{i\in T} x_i + Y_2\right),\ Y_2 \sim \text{Lap}^d(b_T)$\;
Count $L$ users in cluster $k$: $\tilde{N}_{L} \leftarrow |C_k \cap L|$\;
Compute mean of all $L$ users' data in cluster $k$: $\tilde{\mu}_{L} \leftarrow \frac{1}{\tilde{N}_{L}}\sum_{i\in L} \tilde{x}_i$\;
$c \leftarrow \frac{\tilde{N}_{T}}{\tilde{N}_{T} + \tilde{N}_{L}}$,\ $s_T^2 \leftarrow 2b_T^2$,\ and $s_L^2 \leftarrow 2b_L^2$\;
Compute $w_{_{PWH}}$ as defined in Def.~\ref{def:pwh}\;
Update center of $C_k$ to hybrid mean of all data: $w_{_{PWH}}\tilde{\mu}_{T} + (1-w_{_{PWH}})\tilde{\mu}_{L}$\;
}
}
\KwRet{centers of $C_1, \dots, C_k$}
\end{algorithm}

\end{document}